\documentclass[sigconf]{acmart}
\setcopyright{none}

\usepackage[utf8]{inputenc}
\usepackage{graphicx}
\usepackage{latexsym}
\usepackage{amsmath}
\usepackage{amssymb}
\usepackage{url}
\usepackage{float}
\usepackage{tikz}
\usepackage{fancyvrb}
\usepackage{listings}
\usepackage{xcolor}
\usepackage{color}
\usepackage{soul}
\usepackage{multirow}
\usepackage{xspace}
\usepackage{multicol}
\usepackage{subfig}
\usepackage{booktabs}
\usepackage{caption}

\usepackage{paralist}

\copyrightyear{}
\acmYear{}
\setcopyright{none}
\acmConference{}
\acmBooktitle{}
\acmPrice{}
\acmDOI{}
\acmISBN{}

\PassOptionsToPackage{hyphens}{url}\usepackage{hyperref} 
\definecolor{mylinkcolor}{RGB}{0,0,140}
\hypersetup{colorlinks,allcolors=mylinkcolor,citecolor=mylinkcolor}
\usepackage[capitalize]{cleveref}

\usepackage{algorithm}
\usepackage{algpseudocode}
\usepackage{subfig}

\newcommand{\tem}{\tilde{E}^-}
\newcommand{\tep}{\tilde{E}^+}
\newcommand{\teg}{\tilde{G}}

\newcommand{\cut}{\textbf{cut}}
\newcommand{\vol}{\textbf{vol}}

\DeclareMathOperator*{\minimize}{minimize}

\DeclareMathOperator{\argmax}{argmax}

\providecommand{\subjectto}{\ensuremath{\text{subject to}}}

\newcommand{\tony}[1]{}
\renewcommand{\tony}[1]{{\color{ForestGreen}{\bf{Tony says:}} \emph{#1}}}

\title{Parameterized Correlation Clustering \\in Hypergraphs and Bipartite Graphs}

\begin{abstract}
Motivated by applications in community detection and dense subgraph discovery, we consider new clustering objectives in
hypergraphs and bipartite graphs. These objectives are parameterized by one or more \emph{resolution
	parameters} in order to enable diverse knowledge discovery in complex data.

For both hypergraph and bipartite objectives, we identify relevant parameter regimes that are equivalent to
existing objectives and share their (polynomial-time) approximation algorithms. We first show that our
parameterized hypergraph correlation clustering objective is related to higher-order notions of normalized cut and
modularity in hypergraphs. It is further amenable to approximation algorithms via hyperedge expansion techniques.

Our parameterized bipartite correlation clustering objective generalizes standard unweighted bipartite correlation
clustering, as well as the bicluster deletion problem. For a certain choice of parameters it is also related to our
hypergraph objective. Although in general it is NP-hard, we highlight a parameter regime for the bipartite
objective where the problem reduces to the bipartite matching problem and thus can be solved in polynomial time. For
other parameter settings, we present several approximation algorithms using linear program rounding techniques.
These results allow us to introduce the first constant-factor approximation for bicluster deletion, the task of removing a minimum number of edges to partition a bipartite graph into disjoint bi-cliques.

In several experimental results, we highlight the flexibility of our framework and the diversity of results that can be obtained in different parameter settings. This includes clustering bipartite graphs across a range of parameters, detecting motif-rich clusters in an email network and a food web, and forming clusters of retail products in a product review hypergraph, that are highly correlated with known product categories.
\end{abstract}

\author{Nate Veldt}
\affiliation{
  \institution{Cornell University}
  \department{Center for Applied Mathematics}
}
\email{nveldt@cornell.edu}

\author{Anthony Wirth}
\affiliation{%
	\institution{The University of Melbourne}
	\department{Computing and Information Systems}
}
\email{awirth@unimelb.edu.au}

\author{David F. Gleich}
\affiliation{%
  \institution{Purdue University}
 \department{Department of Computer Science}
}
\email{dgleich@purdue.edu}

\begin{document}

\maketitle

\section{Introduction}
Finding sets of related objects in a large dataset, i.e., \emph{clustering}, is one of the fundamental tasks in data mining and machine learning, and is often used as a first step in exploring and understanding a new dataset. When the data to be clustered is represented by a graph or network, the task is referred to as graph clustering or community detection~\cite{Fortunato36, schaeffer2007graphclustering}. A good graph clustering is one in which nodes in the same cluster share many edges with each other, but nodes in different clusters share few edges. While these basic principles are shared by nearly all graph clustering techniques, there are many ways to formalize the notion of a graph cluster~\cite{Fortunato36,schaeffer2007graphclustering,Veldt:2018:CCF:3178876.3186110}. However, no one method or objective function is capable of solving all graph clustering tasks~\cite{Peel2017ground}. 

One outcome is that there are many graph clustering objectives that rely on one or more tunable \emph{resolution}
parameters, which can control the size, structure, or edge density of the clusters formed by optimizing the
objective~\cite{Veldt:2018:CCF:3178876.3186110,Veldt2019learning,Arenas2008analysis,Delvenne2010stabilitypnas,ReichardtBornholdt2004,traag2011narrow}.
In addition to providing a way to detect clusters at different resolutions in a graph, parametric clustering objectives
often make it possible to interpolate between other existing and commonly studied graph clustering objectives. Recently, we showed~\cite{Veldt:2018:CCF:3178876.3186110} that a number of popular graph clustering objectives such as modularity~\cite{newman2004modularity}, normalized cut~\cite{ShiMalik2000}, and cluster deletion~\cite{ShamirSharanTsur2004} can be captured as special cases of a parametric variant of correlation clustering~\cite{BansalBlumChawla2004}.

Nearly all existing techniques for parametric graph clustering focus on a simple graph setting, where all nodes are of the same type and are inter-related by pairwise connections, represented
by edges. However, graph and complex network datasets often have additional structure, which can be exploited for the purpose of more in-depth data analysis. As an example, there has been a recent surge of interest in \emph{higher-order} methods for clustering~\cite{amburg2019hypergraph,BensonGleichLeskovec2016,yin2017local,panli2017inhomogeneous,panli_submodular,Zhou2006learning,hao2018higher,Tsourakakis:2017:SMG:3038912.3052653,veldt2020hypergraph}.
These determine the clustering of the data not only via its graph edges, but also based on its motifs (small, frequently
appearing subgraphs), or indeed based on hyperedges in a hypergraph.
Motifs and hyperedges admit encoding multiway relationships between sets of three or more nodes. This provides a more faithful way to represent complex systems characterized by interactions that are inherently multiway. For example, in co-authorship datasets, papers are frequently written by more than two authors.
Applications of higher-order and hypergraph clustering include image segmentation and computer vision
problems~\cite{Agarwal2005beyond, kim2011highcc}, circuit design and VLSI
layout~\cite{hadley1992efficient,Karypis:1999:MKW:309847.309954}, and bioinformatics~\cite{michoel2012molecular,tian2009gene}.

Bipartite graphs model interactions between two different types of objects.
These have a close relationship with hypergraphs in general, as witnessed
in the example of co-authorship data.
In a hypergraph, each author is a node and the set of authors in each paper
is represented by a hyperedge.
In a bipartite graph, one set of nodes represents authors, the other set papers:
nodes~$i$ and~$j$ are adjacent whenever person~$i$ is an author of paper~$j$.
Which representation is best depends on the task; 
importantly, either of these is more informative than a simple network
in which each edge indicates whether that pair of authors have ever co-authored.

Just as there are many objective functions for graph clustering, many different objectives for clustering hypergraphs and bipartite graphs have been developed, each of which strikes different balances in terms of size and structure of output clusters~\cite{panli2017inhomogeneous,amit2004bicluster,fukunga2018highcc,kim2011highcc,Li2017motifcc,Veldt2018ccgen,Li2019motif,Ailon2011bcc,AsterisKryillidisPapailiopoulosEtAl2016,Zhou2006learning,kaminski2019clustering,ChawlaMakarychevSchrammEtAl2015}. The prevalence and variety of different methods indicates that hypergraphs and bipartite graphs can also exhibit clustering structure at different resolutions. However, these existing methods for clustering hypergraphs and bipartite graphs largely ignore parametric clustering objectives. Thus, in this paper we present a rigorous framework for parametric clustering in these settings. Our objectives are based on parameterized versions of correlation clustering and we show how, in certain parameter regimes, our objectives are related to a number of these previous objectives for bipartite and hypergraph clustering. Furthermore, our methods come with new approximation results. In summary, 
\begin{enumerate}
	\item We present \textsc{HyperLam}, a parametric hypergraph clustering objective that we prove is related to hypergraph generalizations of the normalized cut and modularity objectives.
	\item We present a parametric bipartite correlation clustering objective (PBCC), which captures standard bipartite correlation clustering and bicluster deletion~\cite{amit2004bicluster} as special cases. We also prove that in certain parameter regimes it is equivalent to a variant of our \textsc{HyperLam} objective. 
	\item We prove that \textsc{HyperLam} admits an~$O(\log n)$ approximation by combining certain \emph{expansion} techniques with approximation algorithms for correlation clustering in graphs. We also consider faster heuristic approaches based on applying greedy agglomeration methods.
	\item While PBCC is NP-hard in general, we prove that in a certain parameter regime it is equivalent to bipartite matching and can thus be solved in polynomial time. 
	\item Via linear programming relaxation techniques, we show a number of approximation algorithm that apply to different parameter settings of PBCC, including the first constant factor approximation for bicluster deletion, the problem of partitioning a bipartite graph into disjoint bicliques by removing a minimum number of edges.
\end{enumerate}
As a brief overview of our paper, we begin with small technical preliminaries on correlation clustering, graph
clustering, and hypergraph clustering. Then we state our two new objectives for parametric hypergraph and bipartite
clustering in Sections~\ref{sec:hyperlam}~and~\ref{sec:pbcc}, and prove their equivalence with existing objectives. We discuss algorithms and heuristics in Section~\ref{sec:algs} before showing how these algorithms work in a variety of scenarios (Section~\ref{sec:experiments}).

\section{Preliminaries}
We begin with technical preliminaries on correlation clustering, graph clustering, and hypergraph clustering.

\subsection{Correlation Clustering}
A standard weighted instance of correlation clustering is given by a graph $G = (V, W^+, W^-)$, where each pair of nodes $(i,j) \in V \times V$ with $i \neq j$ is associated with positive and negative weights $w_{ij}^+ \in W^+$ and $w_{ij}^- \in W^-$. Given this input, the objective is to minimize the weight of \emph{mistakes} or \emph{disagreements}. If nodes $i$ and $j$ are clustered together, they incur a mistake with penalty $w_{ij}^-$, and if they are separated, they incur a mistake with penalty $w_{ij}^+$. For instances where at most one of $(w_{ij}^+, w_{ij}^-)$ is non-zero, this can be viewed as a clustering problem in a signed graph. The objective can formally be stated as a binary linear program (BLP):
\begin{equation}
\label{eq:gencc}
\begin{array}{lll} \text{minimize} & \sum_{i < j} w_{ij}^+ x_{ij}^{}  + w_{ij}^- (1-x_{ij}^{})\\  \subjectto
& x_{ij}^{} \leq x_{ik}^{} + x_{jk}^{} \hspace{.5cm} \text{ for all $i,j,k$} \\
& x_{ij}^{} \in \{0,1\} \hspace{1.1cm}  \text{ for all $i < j$}.
\end{array}
\end{equation}
The objective was first presented for signed graphs,
by Bansal et al.~\cite{BansalBlumChawla2004} and by Shamir et al.~\cite{ShamirSharanTsur2004}.
Since its introduction, numerous variations on the objective have been presented for different weighted cases and graph types~\cite{CharikarGuruswamiWirth2005,DemaineEmanuelFiatEtAl2006,Puleo2018correlation,Ailon2011bcc,AilonCharikarNewman2008,Veldt:2018:CCF:3178876.3186110,Li2017motifcc}.
In \emph{bipartite correlation
clustering}~\cite{amit2004bicluster,Ailon2011bcc,ChawlaMakarychevSchrammEtAl2015,AsterisKryillidisPapailiopoulosEtAl2016},
nodes can be organized into two different sets, in such a way that $w_{ij}^+ = w_{ij}^- = 0$ for any pair of nodes~$i$ and~$j$ in the same set. In the complete, unweighted bipartite signed graph case,
the best approximation factor proven
is~3~\cite{ChawlaMakarychevSchrammEtAl2015}.

\subsection{Graph Clustering}
Graph clustering is the task of separating the nodes of a graph into clusters in such a way that nodes inside a cluster
share many edges with each other, but few with the rest of the graph. For an overview of graph clustering and community
detection, we refer to surveys by Fortunato and Hric~\cite{Fortunato36}, and
Schaeffer~\cite{schaeffer2007graphclustering}. Given a graph $G = (V,E)$, we let $\mathcal{C} = \{S_1, S_2, \hdots ,
S_k\}$ represent a disjoint clustering of~$V$, with $S_i \cap S_j = \emptyset$ for $i \neq j$, and $\bigcup_i S_i = V$.
Given a set of nodes~$S \subseteq V$, let $\overline{S} = V\backslash S$ denote the complement set, and~$\cut(S)$ be the
weight of edges between~$S$ and~$\overline{S}$. One of the most common approaches to graph clustering is to set up and solve (or approximate) a combinatorial objective function that encodes some notion of clustering structure. One common objective used for bipartitioning a graph is the normalized cut objective, defined for a set $S \subseteq V$ to be
\begin{equation}
\label{eq:ncut}
{\phi(S) = \frac{\cut(S)}{\vol(S)} + \frac{\cut(S)}{\vol(\overline{S})}} \,,
\end{equation}
where $\vol(S) = \sum_{i \in S} d_i$, with $d_i$ being the degree of node $i$. Another very popular approach is to maximize the modularity objective~\cite{newman2004modularity}, which measures the difference between the number of edges inside a cluster, and the expected number of edges in the cluster, where expectation is defined by some underlying graph null model. 

\paragraph{Flexible parametric frameworks for graph clustering}  Recently, we introduced a framework for graph
clustering based on correlation clustering called LambdaCC~\cite{Veldt:2018:CCF:3178876.3186110}. Given a graph $G =
(V,E)$, the LambdaCC framework replaces an edge $(i,j) \in E$ with a positive edge of weight $1-\lambda d_i d_j$. For
every pair $(i,j) \notin E$, a negative edge of weight $\lambda d_i d_j$ is introduced. The resulting signed graph can then be partitioned with respect to the correlation clustering objective. LambdaCC generalizes several other objectives including normalized cut~\cite{ShiMalik2000}, modularity~\cite{newman2004modularity}, and cluster deletion~\cite{ShamirSharanTsur2004}. 

\subsection{Hypergraph clustering}
We let $\mathcal{H} = (V,\mathcal{E})$ denote a hypergraph, where~$V$ is a set of nodes, and $\mathcal{E}$ is a set of \emph{hyperedges}, which involve two or more nodes. In hypergraphs, the notion of cuts and clustering becomes even more complex, as there can be numerous ways to partition the nodes of a hyperedge, and numerous ways to generalize a graph-based objective. We say that a hyperedge $e \in \mathcal{E}$ is \emph{cut} if it spans
at least two clusters of a clustering,~$\mathcal{C}$. In many clustering applications, any way of separating the nodes of a hyperedge is associated with a penalty equal to the weight of the hyperedge, though other more general notions of hyperedge cuts have also been considered~\cite{gong_km1,Catalyurek99hypergraph,panli2017inhomogeneous,panli_submodular}. Given a set of nodes $S \subseteq V$ in a hypergraph $\mathcal{H}$, we let $\partial S = \{e \in \mathcal{E} : S \cap e \neq \varnothing, \bar{S} \cap e \neq \varnothing \}$ denote the boundary of $S$, and use $\cut_\mathcal{H}(S)$ to denote the hypergraph cut penalty for $S$. The most basic type of cut penalty is to simply count the number of edges on the boundary: $\cut_\mathcal{H}(S) = |\partial S|$. In this paper we also will consider the \emph{linear} cut penalty, defined as follows:
\begin{equation}
\label{eq:2lin}
{ \cut_\mathcal{H}(S) = \sum_{e \in \mathcal{E}} \min \{|S \cap e|, |\bar{S} \cap e| \} }\,.
\end{equation}
Hypergraph generalizations of the normalized cut objective have also been introduced in practice~\cite{panli_submodular,Zhou2006learning,panli2017inhomogeneous}. Here we consider the following definition, first introduced for generalized hypergraph cut functions by Li et al.~\cite{panli2017inhomogeneous}:
\begin{equation}
\label{eq:hncut}
\phi_\mathcal{H}(S) = \frac{\cut_\mathcal{H}(S)}{\vol_\mathcal{H}(S)} +
\frac{\cut_\mathcal{H}(S)}{\vol_\mathcal{H}(\overline{S})}\,,
\end{equation}
where $\cut_\mathcal{H}$ is any hypergraph cut function (e.g., $|\partial S$| or~\eqref{eq:2lin}), and $\vol_\mathcal{H}(S) = \sum_{s \in S} d_s$ is the hypergraph volume of $S$. In this paper we will always consider the hypergraph degree $d_s$ of a node to be the number of hyperedges a node participates in, though other definitions are possible~\cite{panli2017inhomogeneous,panli_submodular}. We also note that hypergraph generalizations of the modularity objective have been considered in different contexts~\cite{kaminski2019clustering,kumar2020new}.

\section{Parametric Hypergraph Clustering} \label{sec:hyperlam}
Our first contribution is a hypergraph clustering objective that differentially treats hyperedges and pairwise edges in a parametric fashion. We further develop equivalence results with existing fixed-parameter objectives;
algorithms are discussed in Section~\ref{sec:algs}.
Given a hypergraph $\mathcal{H} = (V,\mathcal{E})$ and a resolution parameter~$\lambda \in (0,1)$,
we introduce a negative edge between each pair of nodes $(i,j) \in V \times V$, with weight $\lambda w_i w_j$,
where~$w_i$ is a weight associated with node~$i$.
We consider either unit node weights ($w_i = 1$ for all nodes), or degree-based weights:
$w_i = d_i$ for each~$i \in V$.
We treat each original hyperedge in~$\mathcal{H}$ as a positive edge of weight~$1$. In order to accommodate a broad range of possible hyperedge cut penalties, we use the following general abstraction:
let~$P_V$ be the family of all
clusterings, and define $\zeta: \mathcal{E} \times P_V \rightarrow \mathbb{R}$
to be a function that outputs a penalty for the way in which clustering $\mathcal{C} \in P_V$ separates the nodes of a hyperedge $e \in \mathcal{E}$.  The \textsc{HyperLam} objective for a clustering $\mathcal{C}$ of $\mathcal{H}$ is then:
\begin{equation}
\label{eq:hocc}
{ \textsc{HyperLam}(\mathcal{C},\lambda) =  \sum_{e \in \mathcal{E}} \zeta(e, \mathcal{C}) + \sum_{i < j}
	\lambda w_i w_j (1- z_{ij})\,.}
\end{equation}
where 
$z_{ij}$ is a binary indicator for whether nodes $i$ and $j$ are separated ($z_{ij} = 1$) or clustered together $(z_{ij} = 0$) in $\mathcal{C}$. This objective is inspired by the parametric LambdaCC objective for graphs~\cite{Veldt:2018:CCF:3178876.3186110}.

In practice, there may be many meaningful cut functions~$\zeta$ to consider---here we focus mostly on two.
The first is the standard \emph{all-or-nothing} penalty, typically considered in the higher-order correlation clustering literature, which assigns a
penalty proportional to the weight of the hyperedge if and only if
the hyperedge is cut (at least two of its nodes are separated).
Formally, this is defined as
\begin{equation}
\label{eq:allornothing}
{\textstyle \zeta(e,\mathcal{C}) = \begin{cases}
	0 & \text{ if $e \subseteq S$ for some $S \in \mathcal{C}$}\,,\\
	1 & \text{ otherwise.}
	\end{cases}}
\end{equation}
When this standard cut penalty is applied, objective~\eqref{eq:hocc} can be viewed as an instance of higher-order correlation clustering~\cite{kim2011highcc,Veldt2018ccgen,Li2017motifcc,fukunga2018highcc}
with a very special type of negative hyperedge set. Namely, there are no negative hyperedges of size three or more, but \emph{every} pair of nodes defines a negative hyperedge of size two (i.e., a negative edge). The other cut function we consider is a multiway generalization of the linear hypergraph cut penalty~\eqref{eq:2lin}, defined by
\begin{equation}
\label{eq:linear}
\zeta(e,\mathcal{C}) = |e| - \max_{S \in \mathcal{C}} |e \cap S|\,.
\end{equation}
Given a clustering~$\mathcal{C}$, this function assigns a penalty equal to the minimum number of nodes of a hyperedge~$e$ that must be moved in order
for~$e$ to be contained in a single cluster. This has the advantage that for large hyperedges, it assigns a smaller penalty if only a small subset of nodes from a hyperedges are separated from the others. 

\subsection{\textsc{HyperLam} with Linear Cuts}
\label{sec:linear}
If we use the linear hypergraph cut penalty~\eqref{eq:linear}, the \textsc{HyperLam} objective is equivalent to a clustering problem in a bipartite graph obtained by applying a so-called \emph{star expansion}~\cite{zien1999} to $\mathcal{H} = (V,\mathcal{E})$. In more detail, for each $e \in \mathcal{E}$ we can introduce a new node $v_e$, and attach each $v \in e$ to $v_e$ with a unit weight edge. 
Let $V_\mathcal{E}$ be the set of new hyperedge-nodes introduced via this procedure. This defines a graph $G_\mathcal{H} = (\tilde{V}, E)$, where  $\tilde{V} = V \cup V_\mathcal{E}$ is the set of new nodes, and $E$ is the set of edges between $V$ and $V_\mathcal{E}$. 
The goal is then to solve the following objective on $G_\mathcal{H}$:
\begin{equation}
\label{hyperlamlinear}
\minimize \sum_{(i, v_e) \in E} z_{i,v_e} + \lambda\sum_{\substack{i < j\\ i,j \in V}}  w_i w_j (1-z_{ij})\,,\\
\end{equation}
where $z_{i, v_e} = 0$ if node $i\in S$ is clustered with $v_e \in V_\mathcal{E}$, but is one otherwise, and $z_{ij}$ is similarly defined for nodes in $V$. This objective is equivalent to introducing a negative edge of weight $\lambda w_i w_j$ between each pair of nodes in $V$, and optimizing the correlation clustering objective. 

\begin{lemma}
\label{lem:lin}
Objective~\eqref{hyperlamlinear} is equivalent to optimizing \textsc{HyperLam} with the linear cut penalty. 
\end{lemma}
\begin{proof}
Given any fixed clustering $\mathcal{C} = \{{S}_1, \hdots , {S}_k\}$ of the node set $V$, we can define a clustering on all of $G_\mathcal{H}$, by clustering nodes in $V_\mathcal{E}$ in a way that leads to a minimum penalty subject to $\mathcal{C}$. This is accomplished by putting each $v_e \in V_\mathcal{E}$ in the cluster $S = \argmax_{S \in \mathcal{C}} |e \cap S|$. In other words, we put $v_e$ in the cluster where the largest number of nodes in $e$ have been placed, as this minimizes the number of edges adjacent to $v_e$ that are cut. This is the same as applying the linear hypergraph cut penalty~\eqref{eq:linear} to any way of separating nodes in a hyperedge $e \in \mathcal{E}$.
\end{proof}

\subsection{Relationship to Normalized Cut}

Given any hyperedge cut function~$\zeta$, the goal is to optimize~\eqref{eq:hocc} over all possible clusterings of nodes~$V$.
Our first theoretical
result is to show that our new objective captures a hypergraph generalization of normalized cut~\cite{panli2017inhomogeneous}, just as the
LambdaCC graph clustering framework generalizes normalized cut~\cite{Veldt:2018:CCF:3178876.3186110}.
With unit node weights ($w_i = 1$ for all~$i$), Theorem~\ref{thm:ncut} becomes
a statement about a hypergraph variant of the sparsest cut clustering objective. Many aspects of our proof mirror our previous results for the LambdaCC framework~\cite{Veldt:2018:CCF:3178876.3186110}. We have expanded these results to apply to the hypergraph setting. In particular, the second statement regarding the linear hyperedge cut requires significant extra treatment. 
\begin{theorem}
	\label{thm:ncut}
	For degree-weighted \textsc{HyperLam}, there exists some~$\lambda \in (0,1)$, such that optimizing~\eqref{eq:hocc} over biclusterings of the form $\mathcal{C} = \{S, \overline{S}\}$ for some $S \subseteq V$, will produce the minimum hypergraph normalized cut partition~\eqref{eq:hncut}. Furthermore, if the linear penalty~\eqref{eq:linear} is used and we optimize over an arbitrary number of clusters, there exists some $\lambda'$ such that~\eqref{eq:hocc} will be minimized by the minimum hypergraph normalized cut objective under the linear hypergraph cut function~\eqref{eq:2lin}.
\end{theorem}
\begin{proof}
	Let $S^* \subseteq V$ be a set of nodes that minimizes
	\begin{equation}
	\label{eq:ssc}
	\psi(S) = \frac{\cut_\mathcal{H}(S)}{\vol_\mathcal{H}(S) \vol_\mathcal{H}(\overline{S})} \,,
	\end{equation}
	and define $\lambda^* = \psi(S^*)$. Observe that this is simply a scaled version of the hypergraph normalized cut solution: for all $S \subseteq V$, $\psi(S) = \phi(S)/\vol_\mathcal{H}(V)$. Thus, $S^*$ is in fact the optimal hypergraph normalized cut solution as well.
	
	For a bipartition $\mathcal{C} = \{S, \overline{S}\}$, the penalty for positive hyperedges given in~\eqref{eq:hocc} reduces to
	\begin{equation}
	\sum_{e \in \mathcal{E}} \zeta(e, \mathcal{C}) = \cut_\mathcal{H}(S)\,,
	\end{equation}
	where $\cut_\mathcal{H}$ can be any generalized notion of hypergraph cut function~\cite{panli2017inhomogeneous,panli_submodular,veldt2020hypergraph}. 
	When we use degree weights $w_i = d_i$, the second term in objective~\eqref{eq:hocc} can be expressed in terms of the volume of $S$, so that we can write the objective for a bipartition $\mathcal{C} = \{S, \overline{S}\}$ as
	\begin{align}
	\label{eq:hls1}
	\textsc{HyperLam}(S,\lambda) 
	&= \cut_\mathcal{H}(S) - \lambda\sum_{i \in S} \sum_{i \in V \backslash S} d_i d_j + \sum_{i < j} \lambda d_i d_j\\
	\label{eq:hls2}
	&= \cut_\mathcal{H}(S) - \lambda \vol_\mathcal{H}(S) \vol_\mathcal{H}(\overline{S}) + \sum_{i < j} \lambda d_i d_j\,.
	\end{align}
	Since the last term is a constant, if we fix~$\lambda$ and minimize~\eqref{eq:hls2}, we can check whether the minimizer $S$ satisfies:
	\begin{equation*}
	\cut_\mathcal{H}(S) - \lambda \vol_\mathcal{H}(S) \vol_\mathcal{H}(\overline{S})  < 0 \implies \frac{\cut_\mathcal{H}(S)}{\vol_\mathcal{H}(S) \vol_\mathcal{H}(\overline{S})} < \lambda \,.
	\end{equation*}
	This means that minimizing the \textsc{HyperLam} objective over bipartitions is equivalent to solving the decision version of hypergraph normalized cut, i.e., given a fixed $\lambda$, find whether there is any bipartition $\{S, \overline{S}\}$ such that $\psi(S) < \lambda \implies \phi(S) < \vol_\mathcal{H}(S) \lambda$. Thus, for a small enough value of $\lambda$, which will be slightly larger than $\lambda^*$, the optimal solutions to \textsc{HyperLam} and $\phi$ will coincide. 
	
%

\paragraph{Using the linear cut penalty}
Next we prove that for the linear cut penalty~\eqref{eq:linear}, the \textsc{HyperLam} objective generalizes hypergraph normalized cut objective with linear penalty~\eqref{eq:2lin}, even if we do not restrict to considering only bipartitions of $V$. To prove this, we use the characterization of the \textsc{HyperLam} objective given in Section~\ref{sec:linear}. Specifically, optimizing \textsc{HyperLam} with the linear cut penalty on a hypergraph $\mathcal{H} = (V,\mathcal{E})$ is equivalent to optimizing objective~\eqref{hyperlamlinear} on a bipartite graph $G_\mathcal{H} = (V\cup V_\mathcal{E}, E)$. For this relationship, every clustering $\mathcal{C}$ of $V$ induces a clustering $\tilde{\mathcal{C}}$ of $\tilde{V} = V \cup V_\mathcal{E}$, obtained by arranging nodes of $V_\mathcal{E}$ in a way that minimizes cut edges between $V$ and $V_\mathcal{E}$. We will call $\tilde{\mathcal{C}}$ the \emph{natural extension} of $\mathcal{C}$. Similarly, for $\tilde{S} \in \tilde{\mathcal{C}}$ and $S = \tilde{S} \cap V$, we call $\tilde{S}$ the natural extension of $S$ in $\tilde{\mathcal{C}}$.


For any two disjoint subsets $\tilde{S} \subseteq \tilde{V}$ and $\tilde{T} \subseteq \tilde{V}$ with $\tilde{S} \cap \tilde{T} = \varnothing$, let $\cut(\tilde{S}, \tilde{T})$ denote the number of edges between these two sets in $G_\mathcal{H}$, and define $\cut(\tilde{S}) = \cut(\tilde{S}, \tilde{V} \backslash\tilde{S})$. The \textsc{HyperLam} objective (with linear penalty) for this clustering is given by
\begin{equation}
\label{eq:hyperlincut}
\frac12 \sum_{\tilde{S} \in \tilde{\mathcal{C}}} \cut(\tilde{S}) + \sum_{\substack{i < j\\ i,j \in V}} \lambda d_i d_j - \frac{\lambda}{2} \sum_{{S} \in {\mathcal{C}}} \vol_\mathcal{H}(S) \vol_\mathcal{H}(V \backslash S)\,.
\end{equation}
We can interpret~\eqref{eq:hyperlincut} in terms of positive and negative edge mistakes in the underlying instance of correlation clustering. The first term corresponds to positive edge mistakes, made by cutting edges between $V$ and $V_\mathcal{E}$. The second term is the weight of all negative edges, and the third term subtracts the weight of negative edges between two different clusters, since these are the negative edges where the clustering \emph{does not} make a negative edge mistake. The first and third terms are multiplied by $1/2$, since these terms account for edges that are adjacent to exactly two clusters.

Mirroring our proof for the bipartition case, we can see that for a fixed $\lambda$, minimizing~\eqref{eq:hyperlincut} over arbitrary clusterings allows us to check whether there exists a clustering ${\mathcal{C}}$ such that
\begin{equation}
\Psi({\mathcal{C}}) = \frac{\sum_{\tilde{S} \in \tilde{\mathcal{C}}} \cut(\tilde{S}) }{\sum_{{S} \in {\mathcal{C}}} \vol_\mathcal{H}(S) \vol_\mathcal{H}(V \backslash S)} < \lambda \,,
\end{equation}
where $\tilde{\mathcal{C}}$ is the natural extension of $\mathcal{C}$. Therefore, for a certain value of $\lambda$, the solution to \textsc{HyperLam} with the linear hyperedge cut penalty will be the same as the minimizer for $\Psi$. Observe that for a bipartition $\mathcal{C} = \{S, V\backslash S\}$, $\Psi(\mathcal{C}) = \psi(S)$. Although $\Psi$ is defined for clusterings of arbitrary size, we will show that it is minimized by a bipartition. First, we prove a relationship between the cut function in $G_\mathcal{H} = (V\cup V_\mathcal{E}, E)$ (denoted by $\cut$) and the two-way linear cut function~\eqref{eq:2lin} in the original hypergraph $\mathcal{H} = (V, \mathcal{E})$ (denoted by $\cut_\mathcal{H}$). Let $S$ be a cluster in an arbitrary clustering $\mathcal{C}$ (not necessarily a bipartition), and let $\tilde{S}$ and $\tilde{\mathcal{C}}$ be their natural extensions so that $S = \tilde{S}\cap V$. Then
\begin{align*}
\cut(\tilde{S}) &=  \cut(S, (\tilde{V}\backslash \tilde{S}) \cap V_\mathcal{E}) + \cut(\tilde{S} \cap V_\mathcal{E}, V \backslash S)\\
&= \sum_{v_e \in (\tilde{V}\backslash \tilde{S}) \cap V_\mathcal{E}} \cut(S, \{v_e\}) + \sum_{v_e \in \tilde{S} \cap V_\mathcal{E}} \cut(V\backslash S, \{v_e\})\\
&= \sum_{e \in \mathcal{E}} |S \cap e| + \sum_{e \in \mathcal{E}} |V\backslash S \cap e|\\
&\geq \sum_{e \in \mathcal{E}} \min \{ |S\cap e| , |V\backslash S \cap e|\}= \cut_\mathcal{H}(S)\,.
\end{align*}

Finally, let $\mathcal{C}'$ be an arbitrary minimizer for $\Psi$. Then, 
\begin{align}
\Psi({\mathcal{C}'}) &=\frac{\sum_{\tilde{S} \in \tilde{\mathcal{C}'}} \cut(\tilde{S}) }{\sum_{{S} \in {\mathcal{C}'}} \vol_\mathcal{H}(S) \vol_\mathcal{H}(V \backslash S)}\\
& \geq \min_{S \in\mathcal{C}'}\,\, \frac{\cut(\tilde{S}) }{ \vol_\mathcal{H}(S) \vol_\mathcal{H}(V \backslash S)} \\
&\geq  \frac{\cut_\mathcal{H}(S^*) }{ \vol_\mathcal{H}(S^*) \vol_\mathcal{H}(V \backslash S^*)}\\
&= \psi(S^*) = \Psi(\mathcal{C}^*).
\end{align}
This confirms that for a certain value of $\lambda$, $\mathcal{C}^* = \{S^*, V\backslash S^*\}$ is optimal for \textsc{HyperLam} with a linear hyperedge cut penalty.

\end{proof}

\section{Parametric Bipartite Clustering} \label{sec:pbcc}
Next we present a parameterized variant of bipartite correlation clustering (PBCC) in graphs, which we prove generalizes a number of other clustering objectives
in bipartite graphs, and comes with several novel approximation guarantees. Let $G = (V_1, V_2, E)$ be a bipartite graph in which $V_1$ and $V_2$ are node sets and $E$ is a set of edges between nodes in $V_1$ and $V_2$. In order to define an instance of Parametric Bipartite Correlation Clustering (PBCC), we first define parameters $\mu_1$, $\mu_2$, and $\beta$, all in the interval $[0,1]$. We then associate each $e \in E$ with a positive edge of weight $1-\beta$, and every $e \in (V_1 \times V_2) - \{E\}$ with a negative edge with weight $\beta$. Additionally, each pair of nodes in~$V_1$ is given a negative edge of weight~$\mu_1$, and
each pair of nodes in $V_2$ is given a negative edge of weight~$\mu_2$. The result is a complete, weighted instance of correlation clustering, where the underlying positive edge structure is a bipartite graph. We illustrate an instance of the problem in Figure~\ref{fig:pbcc}. 
Our PBCC objective is
\begin{equation}
\label{eq:pbcc}
\begin{split}
\textsc{PBCC}(\mathcal{C}) = 
& {\sum_{i \in V_1, \, j \in V_2} [\beta (1-A_{ij})(1-z_{ij})  + (1-\beta) A_{ij}z_{ij} ]} \\
&+ \textstyle{\sum_{(i,j) \in V_1 \times V_1} \mu_1 (1-z_{ij})  + \sum_{(i,j) \in V_2 \times V_2} \mu_2 (1-z_{ij})  }\,,
\end{split}
\end{equation}
where $A_{ij} = 1$ if $(i,j) \in E$, but is zero otherwise, and $z_{ij}$ is the indicator for node separation in $\mathcal{C}$ ($z_{ij} = 0$ means $(i,j)$ are clustered together) as before.
\begin{figure}[t]
	\centering
	\includegraphics[width=.5\linewidth]{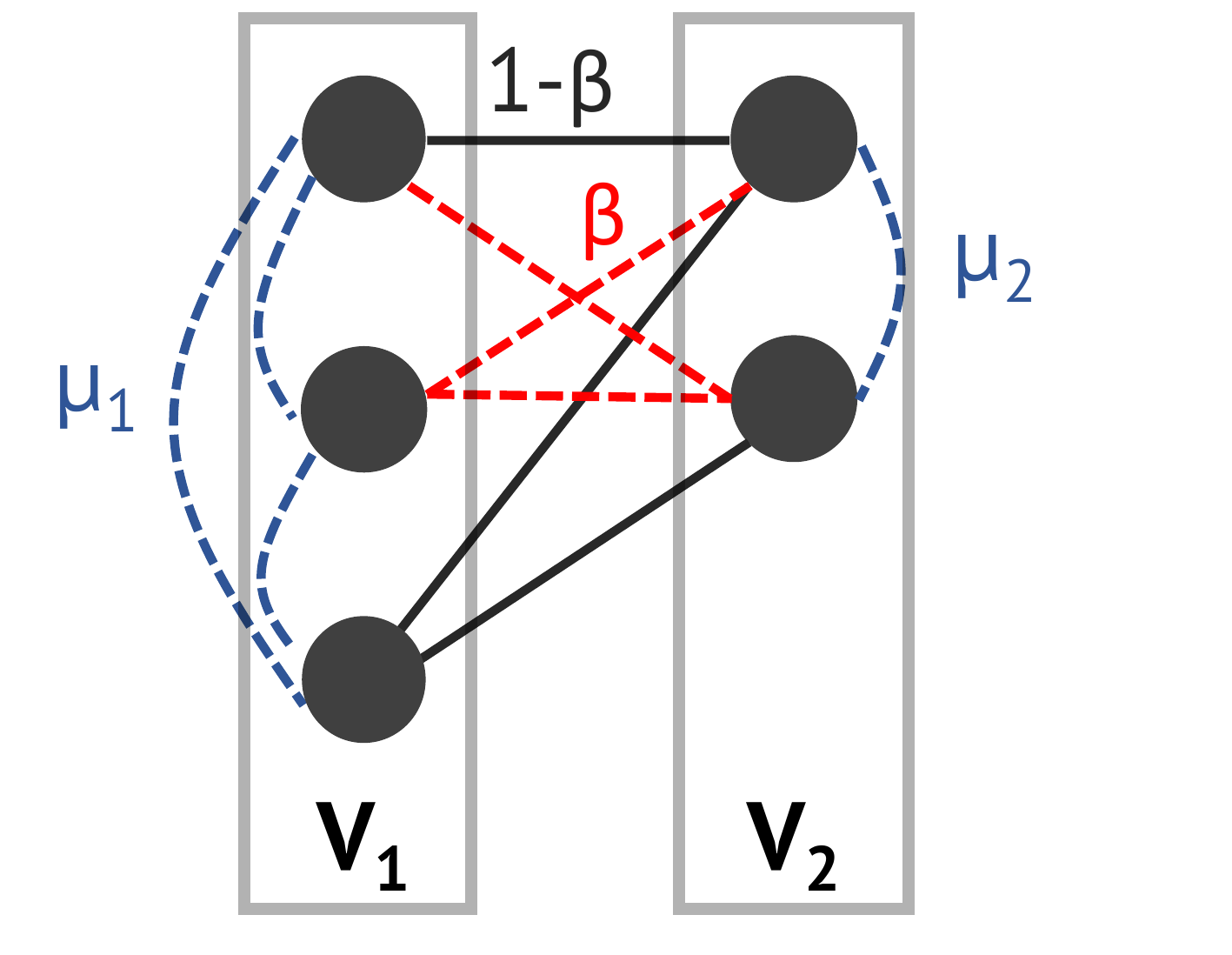}
	\caption{Parameterized BCC is given by a complete signed graph with edge weighted parameterized by $\mu_1$, $\mu_2$ and $\beta$. Edges of weight $\beta$ correspond to missing edges in some underlying bipartite graph $G = (V_1, V_2, E)$. 
	}
	\label{fig:pbcc}
\end{figure}
\begin{table}
	\centering
	\caption{Equivalence and approximation results for PBCC; $\varepsilon$ represents a small, graph dependent number.}
	\begin{tabular}{llll}
		\toprule
		 \textbf{Parameters} & \textbf{Equivalence} & \textbf{Approx.} \\
		\midrule
		$\beta = \mu_1 = \mu_2 = \lambda$ & LambdaCC & see~\cite{Veldt:2018:CCF:3178876.3186110}, \cite{Veldt2018ccgen}\\
		\midrule
		$\mu_1 = \mu_2 \geq (1-\beta)$ & Bip. Matching & 1  (Thm~\ref{thm:matching}) \\		
		\midrule
		$\mu_1 = \mu_2 = 0$, $\beta \geq 1-\varepsilon$ & Bicluster deletion & 4 (Thm~\ref{thm:4bicluster})\\
		\midrule
		$\mu_1 = \mu_2 = 0$, $\beta \geq \frac{1}{2}$ & Generalized BCC & $6-\frac{1}{\beta}$ (Thm~\ref{thm:5orbetter})\\
		\midrule
		$\mu_1 = \mu_2\in [0,1]$, $\beta \geq \frac{1}{2}$ & - & 5 (Thm~\ref{thm:5app})\\
		\midrule
		$\mu_1 = \lambda$, $\mu_2 = 0$, $\beta = 0$ & \textsc{HyperLam} & $O(\log n)$ \\
		\bottomrule
	\end{tabular}
	\label{tab:pbcc}
\end{table}
The PBCC objective is closely related to several other well-studied problems. We summarize a list of equivalence results and approximation algorithms for PBCC in Table~\ref{tab:pbcc}, based on the results we prove in this section and the next.

 When $\mu_1 = \mu_2 = 0$ and $\beta = 1/2$, the problem corresponds to the standard unweighted bipartite correlation clustering problem (BCC)~\cite{Ailon2011bcc,amit2004bicluster}. When $\mu_1 = \mu_2 = \beta$, it is equivalent to applying the LambdaCC framework~\cite{Veldt:2018:CCF:3178876.3186110} to a bipartite graph.

If $\beta > |V_1||V_2| /(|V_1||V_2| + 1)$ and $\mu_1 = \mu_2 = 0$, then making a mistake at a single negative edge of weight $\beta$ introduces a
greater weight of disagreements than placing each node into a singleton cluster. Therefore, the objective will be optimized by making a minimum number of positive-edge mistakes, subject to all clusters being bicliques. Thus, in this parameter regime, PBCC is equivalent to bicluster deletion, the problem of removing a minimum number of edges from a bipartite graph to partition it into disjoint bicliques. 

\subsection{Relationship with Bipartite Matching}
Although PBCC is NP-hard in general, our next theorem shows that in a certain parameter regime, PBCC is equivalent to solving a bipartite matching problem on $G = (V_1, V_2, E)$. Therefore, the problem can be solved in polynomial time in this regime.
\begin{theorem}
	\label{thm:matching}
	If parameters $\mu_1$, $\mu_2$, and $\beta$ satisfy $\min \{\mu_1, \mu_2\} \geq (1-\beta)$, then the optimal solution to PBCC for these parameters is the same as finding a maximum bipartite matching on $G = (V_1, V_2,E)$.
\end{theorem}
\begin{proof}
	First consider $\mu = \mu_1 = \mu_2 > 1-\beta$, and let $\mathcal{C}$ denote the optimal clustering in this case. Let $S = \{S_1 \cup S_2\}$ be an arbitrary cluster in $\mathcal{C}$, where $S_i \subseteq V_i$ for $i = 1,2$. Assume without loss of generality that $|S_1| \leq |S_2|$. In three steps, we will prove that $S$ contains at most one node from each of $V_1$ and $V_2$, and thus $\mathcal{C}$ is in fact just a matching.
	
	\textbf{Step 1.} Observe that \textbf{$S$ must be a biclique} in terms of positive edges between $S_1$ and $S_2$. If we assume not, then there exists a node $s \in S_2$ that does not share a positive edge with every node in $S_1$. By removing $s$ from $S$, we no longer make negative edge mistakes between $s$ and the rest of $S_2$, which decreases the objective by $\mu (|S_2| - 1)$. At the same time, this introduces new errors weighing \emph{at most} $(1-\beta) (|S_1| - 1)$, due to positive edge mistakes between $s$ and $S_1$. This \emph{decreases} the overall objective score by at least
	\[ \mu (|S_2| - 1) - (1-\beta) (|S_1| - 1) > 0,\]
	since $|S_2| > |S_1|$ and $(1-\beta) < \mu$. In other words, the weight of mistakes strictly decreases if we removed $s$. This would be a contradiction to the optimality of $\mathcal{C}$. Thus, no such $s$ exists, and $S$ is a biclique of positive edges.
	
	\textbf{Step 2.} If $|S| > 1$, then $|S_1| = |S_2|$. If we assume instead that $|S_2| \geq |S_1| +1$, then removing any $s \in S_2$ would decrease the objective by $\mu (|S_2|-1)$ and increase the objective by $(1-\beta)|S_1|$, since $S$ is a biclique. Since $|S_1| \leq |S_2| -1$, this again leads to an overall decrease in the weight of mistakes:
	\[ \mu (|S_2| - 1) - (1-\beta) |S_1|  \geq \mu (|S_2| - 1) - (1-\beta) (|S_2|-1) > 0,\]	
	so a contradiction is shown.
	
	\textbf{Step 3.} If $|S| > 1$, then $|S_1| = |S_2| = 1$. By Step 2, $S_1$ and $S_2$ have the same size $k$. We will prove that $k = 1$. By Step 1, every node in $S_1$ shares a positive edge with every node in $S_2$. Note that this implies there is a perfect matching between the two sides. Thus, we can split up $S$ into multiple subclusters where each node in $S_1$ is paired up with a single positive neighbor in $S_2$. This removes the negative edge mistakes on both sides of the graph, improving the objective by $\mu k(k-1)$. In turn, it introduces positive edge mistakes weighing $(1-\beta)k (k-1)$, since each node on one side now only is clustered with a single node on the other side. However, this change improves the overall objective score, since $\mu > (1-\beta)$. Therefore, the only way for there not be a contradiction is if $k = 1$ to begin. 
	
	The above steps show that when $\mu > (1-\beta)$, every cluster in $\mathcal{C}$ is either a singleton, or it contains exactly one node from $V_1$ and one node from $V_2$. Therefore, $\mathcal{C}$ is a matching, and in fact it will be a maximum matching in order to make as few mistakes for the correlation clustering as possible. Observe now that if $\mu = (1-\beta)$, even if there exists an optimal clustering that is not a matching, we can use the above arguments to show we can break the clustering into a matching without making the objective score worse. Finally, if $\mu_1 \neq \mu_2$ but $\mu = \min \{\mu_1, \mu_2\} \geq 1-\beta$, then having one $\mu_i > \mu$ simply introduces \emph{more} repulsion between nodes due to negative weights. Thus, the arguments above still hold, and a maximum matching is still optimal.
\end{proof}

\subsection{Relationship with \textsc{HyperLam}}
When $\mu_1 = \lambda$, $\mu_2 = 0$, and $\beta = 0$, PBCC is equivalent to a special instance of \textsc{HyperLam} with a linear hyperedge cut penalty~\eqref{eq:linear}. As we noted in Section~\ref{sec:linear}, when we use this penalty, the \textsc{HyperLam} objective is equivalent to an instance of correlation clustering defined by performing a star expansion. This results in an instance of PBCC where $V_1 = V$ is the set of original nodes, each pair of which has a negative edge of weight $\mu_1 =
\lambda$. The auxiliary nodes in $V_\mathcal{E}$ constitute~$V_2$, with $\mu_2 = 0$, and edges between $V$ and $V_2$ all have weight $1-\beta = 1$.
%
%

\section{Approximations and Heuristics}
\label{sec:algs}
We now turn our attention to specific approximation guarantees that can be obtained by our objectives in different parameter regimes. We begin by
reviewing a general strategy for obtaining approximation guarantees for variants of correlation clustering, through which we prove approximation
guarantees for PBCC. In order to approximate \textsc{HyperLam}, we combine existing approximation algorithms for correlation clustering with
techniques for reducing a hypergraph to a related graph. We conclude with some heuristic approaches for \textsc{HyperLam}.

\subsection{General LP Rounding Algorithm for CC}
	\begin{algorithm}[tb]
		\caption{Pivot}
		\begin{algorithmic}[5]
			\State{\bfseries Input:} Unweighted signed graph $G = (V,E^+,E^-)$
			\State {\bfseries Output:} Clustering~$\mathcal{C} =\textsc{Pivot}(G)$
			\State Select a pivot node $k \in V$
			\State Form cluster $S = \{v \in V: (k,v) \in E^+ \}$
			\State Output clustering $\mathcal{C} = \{S, \textsc{Pivot}(G\backslash S) \}$
		\end{algorithmic}
		\label{alg:piv}
	\end{algorithm}
	
	\textsc{Pivot} (aka Algorithm~\ref{alg:piv}) is a simple algorithm unweighted correlation clustering.
When pivots are chosen uniformly at random, Ailon et al.~\cite{AilonCharikarNewman2008} showed that this algorithm returns a 3-approximation for complete unweighted correlation clustering. Later, van Zuylen and Williamson~\cite{vanzuylen2009deterministic} produced a de-randomized 3-approximation. 
	We review a generic algorithm for correlation clustering from the work of van Zuylen and Williamson, which we apply later in developing approximation algorithms for new parametric correlation clustering variants. Pseudocode for this method, which we call \textsc{GenRound}, is given in Algorithm~\ref{alg:genround}. 
	\begin{algorithm}[tb]
		\caption{\textsc{GenRound}}
		\begin{algorithmic}[5]
			\State{\bfseries Input:} CC instance $G = (V,W^+,W^-)$, parameter $\delta \in [0,1]$.
			\State {\bfseries Output:} Clustering~$\mathcal{C}$ of $G$.
			\State Solve LP-relaxation of~\eqref{eq:gencc} to obtain \emph{distances} $x_{ij}$, for each $i \neq j$.
			\State 
			\(\tep \gets \{(i,j) : x_{ij} < \delta \}, \hspace{.5cm} \tem \gets \{(i,j): x_{ij} \geq \delta \}  \)
			\State Apply \textsc{Pivot} to $\tilde{G}=(V,\tep, \tem)$.
		\end{algorithmic}
		\label{alg:genround}
	\end{algorithm}
	One can prove approximation results for special input problems using the following theorem. We have adapted this result from the work of van Zuylen and Williamson~\cite{vanzuylen2009deterministic} to match our notation and presentation. The theorem and its proof rely on a careful consideration of so-called \emph{bad triangles} in the rounded unweighted graph. A bad triangle is a triplet of nodes in the graph which contains two positive edges and one negative edge. 
		\begin{theorem}
		\label{thm:3pt1} 
		(Theorem 3.1 in~\cite{vanzuylen2009deterministic}).
		Given a weighted instance of correlation clustering $G = (V, W^+, W^-)$, let $c_{ij}^{} = w_{ij}^+ x_{ij}^{} +
		w_{ij}^-(1-x_{ij}^{})$. \textsc{GenRound} returns an $\alpha$-approximation for the min-disagree objective~\eqref{eq:gencc} if the threshold
		parameter,~$\delta$,
		is chosen so that the graph $\tilde{G} = (V,\tilde{E}^+, \tilde{E}^-)$ satisfies the following conditions:
		\begin{enumerate}
			\item For all~$(i,j) \in \tilde{E}^+$, we
			have~$w_{ij}^- \leq \alpha c_{ij}^{}$, and for all $(i,j) \in \tem$, we have $w_{ij}^+ \leq \alpha c_{ij}^{}$.
			\item For every triangle $(i,j,k)$ in $\tilde{G}$, with $\{(i,j), (j,k)\} \subseteq \tilde{E}^+$ and $(i,k) \in \tilde{E}^-$, we have $w_{ij}^+ + w_{jk}^+ + w_{ik}^- \leq \alpha \left(   c_{ij}^{} + c_{jk}^{} + c_{ik}^{} \right)$.
		\end{enumerate}
	\end{theorem}
When applying \textsc{Pivot} in Algorithm~\ref{alg:genround}, selecting the pivot node uniformly at random gives an \emph{expected} $\alpha$-approximation. A deterministic algorithm with the same approximation factor $\alpha$ can be obtained via a careful selection of pivot nodes~\cite{vanzuylen2009deterministic}.

\subsection{Graph Reductions \mbox{for HyperLam}}
\label{sec:reduction}
Although \textsc{HyperLam} is NP-hard to optimize, we can obtain approximation algorithms for the objective using two different techniques for converting hypergraphs to graphs. 

\textbf{Weighted clique expansion:} Replace each hyperedge $e \in \mathcal{E}$ with a clique on $e$ where each edge has weight $1/(|e| - 1)$. If two
nodes appear together in multiple hyperedges, assign a weight equal to the sum of weights from each such clique expansion.

\textbf{Star expansion:} As outlined in Section~\ref{sec:linear}, replace each hyperedge $e \in \mathcal{E}$ with an auxiliary node $v_e$ and an edge from each $v \in e$ to $v_e$. If we use weights $w_i =1$ for all $i \in V$, this is equivalent to an instance of PBCC with $\mu_1 = \lambda$, $\mu_2 = 0$, and $\beta = 0$.

For each expansion technique, we still include a negative edge of weight $\lambda w_i w_j$ between each pair $(i,j) \in V \times V$, where $w_i$ is the weight for node $i$. Either way, the result is an instance of weighted correlation clustering, which we can solve with existing approximation algorithms.


The weighting scheme for the clique expansion is chosen specifically to approximately model the all-or-nothing hyperedge cut penalty~\eqref{eq:allornothing}. For
three-uniform hypergraphs, the relationship is exact~\cite{ihler1993modeling}. For a $k$-node hyperedge, with $k > 3$, the minimum penalty for
splitting the clique comes from placing all but one node in the same cluster, giving a penalty equal to $(k-1)/(k-1) = 1$. The maximum possible
penalty, when all~$k$ nodes in~$e$ are placed in different clusters, is ${k \choose 2}\frac{1}{k-1} = \frac{k}{2}$.
Thus, the penalty at each positive hyperedge in the resulting reduced graph will be within a factor $k/2$ of the original all-or-nothing penalty for any clustering $\mathcal{C}$. Meanwhile, the star expansion enables us to \emph{exactly} model the linear cut penalty~\eqref{eq:linear}, as shown in Lemma~\ref{lem:lin}. 

Thus, applying existing approximation algorithms for correlation clustering~\cite{CharikarGuruswamiWirth2005,DemaineEmanuelFiatEtAl2006}, we get an
$O(k \log n)$ approximation for \textsc{HyperLam} with all-or-nothing penalty via the weighted clique expansion, where $k$ is the maximum size hyperedge. We also obtain an $O(\log n)$ approximation for
\textsc{HyperLam} with linear hyperedge penalty via the star expansion.


\subsection{A Four-Approx for Bicluster Deletion}
We now show how \textsc{GenRound} and Theorem~\ref{thm:3pt1} combine to develop a $4$-approximation for bicluster deletion:
the first constant-factor approximation for this problem. Rather than the edge weights presented in the last section, we
view bicluster deletion as a general weighted correlation clustering problem with the following weights
%
%
\begin{equation}
(w_{ij}^+, w_{ij}^-) = 
\begin{cases} 
(0, 0) & \text{ if $i$ and $j$ are in the same bipartition of $G$} \\
(1, 0) & \text{ if $(i,j) \in E^+$} \\
(0,\infty) &  \text{ if $(i,j) \in E^-$}.
\end{cases}
\end{equation}
Above, $E^+$ and $E^-$ denote positive and negative edges between the two sides of the bipartite graph. To ensure no mistakes are made at negative edges, we add the constraint $x_{ij} = 1$ to BLP~\eqref{eq:gencc}, for every $(i,j) \in E^-$. The LP-relaxation of this problem is given by
\begin{equation}
\label{eq:bcdlp}
\begin{array}{lll} \text{minimize} & \sum_{(i,j) \in E^+} x_{ij} \\  \subjectto
& x_{ij} = 1 \hspace{1.85cm} \text{ for all $(i,j) \in E^-$} \\
& x_{ij} \leq x_{ik} + x_{jk} \hspace{.5cm} \text{ for all $i,j,k$} \\
& 0 \leq x_{ij} \leq 1 \hspace{1.1cm}  \text{ for all $i < j$}.
\end{array}
\end{equation}
\begin{theorem}
	\label{thm:4bicluster}
	Applying \textsc{GenRound} to LP~\eqref{eq:bcdlp}, with~$\delta = 1/2$, returns a $4$-approximation to bicluster deletion.
\end{theorem}
\begin{proof}
	First of all note that \textsc{GenRound} applied to LP~\eqref{eq:bcdlp} with $\delta = 1/2$ will indeed form only complete bicliques. Applying a pivot step around node $k$ will form a cluster $S$ in which $x_{ki}^{} < \delta = 1/2$ for every $i \in S$. For any two non-pivot nodes $i$ and $j$ in $S$, $x_{ij}^{} \leq x_{ki}^{} + x_{kj}^{} < 1$. This means that $(i,j) \notin E^-$, since the LP relaxation forces all negative edges to have distance one. It remains to check that the conditions of Theorem~\ref{thm:3pt1} are satisfied for $\alpha = 4$. 
	
	For the first condition, if $(i,j) \in \tep$, this means $x_{ij}^{} < 1/2 \implies (i,j) \in E^+$, which means $w_{ij}^- = 0$ so the inequality $w_{ij}^- \leq \alpha c_{ij}$ is trivially satisfied. If $(i,j) \in \tem \cap E^+$, then $w_{ij}^+ = 1$ and $c_{ij}^{} = x_{ij}^{}$, so 
	\[ w_{ij}^+ = 1 < 2 = 4 (1/2) \leq 4 x_{ij}^{} = \alpha x_{ij}^{}. \]
	If $(i,j) \in \tem \cap E^-$, then $w_{ij}^+ = 0$ and again the inequality is trivial. Thus, the first condition is satisfied for all cases.
	
	For the second condition, consider a bad triangle $(i,j,k)$ in $\teg$ with $(i,k) \in \tem \implies x_{ik}^{} \geq 1/2$. Since $(i,j)$ and $(j,k)$ are in $\tep$, $x_{ij}^{} < 1/2$ and $x_{jk}^{} < 1/2$, so $x_{ik}^{} \leq x_{ij}^{} + x_{jk}^{} < 1$. If $(i,j,k)$ are all on the same side of the graph in $G$, then $w_{ij}^+ + w_{jk}^+ + w_{ik}^- = 0$ and the inequality in condition two of Theorem~\ref{thm:3pt1} is trivial. If $i$ and $j$ are on the same side, but not $k$, then
	\[\alpha (c_{ij}^{} + c_{jk}^{} + c_{ik}^{}) = 4 (0 + x_{jk}^{} + x_{ik}^{})\geq 4x_{ik}^{} \geq 2 > 1 =  w_{ij}^+ + w_{jk}^+ + w_{ik}^-.\]
	If~$i$ and~$k$ are on the same side of the graph but not $j$ (which is symmetric to considering~$j,k$ together and~$i$ on the other side), then $w_{ij}^+ + w_{jk}^+ + w_{ik}^- = 2$ and
	\[	\alpha (c_{ij} + c_{jk} + c_{ik}) = 4 (x_{ij}^{} + x_{jk}^{} + 0)\geq 4x_{ik}^{} \geq 2 = w_{ij}^+ + w_{jk}^+ + w_{ik}^-.\]
	Since all the conditions of Theorem~\ref{thm:3pt1} hold in all cases, \textsc{GenRound} is a 4-approximation for bicluster deletion when~$\delta = 1/2$.
\end{proof}

\subsection{Generalized Results for PBCC}
We now turn to approximation algorithms for a wider range of parameter settings. In the remainder of the section, we specifically consider $\mu =
\mu_1 = \mu_2$. As we did for bicluster deletion, our goal is to find a threshold parameter $\delta$ and an approximation factor $\alpha$ such that
the two conditions of Theorem~\ref{thm:3pt1} hold. 
To find the best choice of $\delta$ in different settings, we first set up a system of inequalities that are \emph{sufficient} to guarantee the assumptions of Theorem~\ref{thm:3pt1}. In these inequalities, $\mu$ and $\beta$ are treated as constants, and $\alpha$ and $\delta$ are variables we optimize over to obtain the best approximation results. We wish to find $\delta$ and $\alpha$ such that these sufficient constraints are satisfied and the approximation factor $\alpha$ is minimized. 

\paragraph{Sufficient constraints for first condition.}
The first condition of Theorem~\ref{thm:3pt1} requires that for all $(i,j) \in \tilde{E}^+$, we have $w_{ij}^- \leq \alpha c_{ij}^{}$, and for all $(i,j) \in \tem$, we have $w_{ij}^+ \leq \alpha c_{ij}^{}$. If $(i,j) \in E^+ \cap \tep$ or $(i,j) \in E^- \cap \tem$, then the left hand side of the inequality is zero and the inequality is trivially satisfied.

If $(i,j) \in E^+ \cap \tem$, then $x_{ij}^{} \geq \delta$, $w_{ij}^+ = (1-\beta)$, and $c_{ij}^{} = (1-\beta)x_{ij}^{} \geq (1-\beta) \delta$. Thus the inequality is satisfied as long as
\begin{equation}
\label{cond1a}
\alpha\delta \geq 1\,.
\end{equation}
On the other hand, if $(i,j) \in E^- \cap \tep$, then $x_{ij}^{} < \delta$, $w_{ij}^-$ is either $\mu$ or $\beta$, and $c_{ij}^{} = w_{ij}^-(1-x_{ij}^{}) \geq w_{ij}^- (1-\delta)$. In order for the inequality $w_{ij}^- \leq \alpha c_{ij}^{}$, is it sufficient to choose $\alpha$ and $\delta$ satisfying:
\begin{equation}
\label{cond1b}
\alpha(1-\delta) \geq 1\,.
\end{equation}

\paragraph{Sufficient constraints for second condition.}
The second condition is defined for a triangle $(i,j,k)$ with $(i,k) \in \tem$ and $(i,j), (j,k) \in \tep$. We refer to this as a ``bad triangle,'' since at least one of the edges must be violated by any clustering. The requirement is:
\begin{equation}
\label{eq:triplet}
L = w_{ij}^+ + w_{jk}^+ + w_{ik}^- \leq \alpha (  c_{ij}^{} + c_{jk}^{} + c_{ik}^{} ) = R.
\end{equation}

Following the approach of Ailon et al.~\cite{Ailon2011bcc} for standard BCC, and our 4-approximation for bicluster deletion, the analysis is split into three cases:
\begin{itemize}
	\item Case 1: $\{i,j\}$ are on the same side of $G$, but not $k$.
	\item Case 2: $\{i,k\}$ are on the same side of $G$, but not $k$.
	\item Case 3: $i$, $j$, and $k$ are all on the same side of $G$.
\end{itemize} 
For Case 3, we know all edges in the triangle are negative in $G$, so there are no subcases to consider. However, Case 1 and Case 2 both require we consider four different subcases, since they both involve two edges crossing from one side to $G$ to the other, which could be positive or negative. Figure~\ref{fig:trianglecases} illustrates all the ways a triangle in $G$ can be mapped to a bad triangle in $\tilde{G}$. 
\begin{figure}[t]
	\centering
	\includegraphics[width=\linewidth]{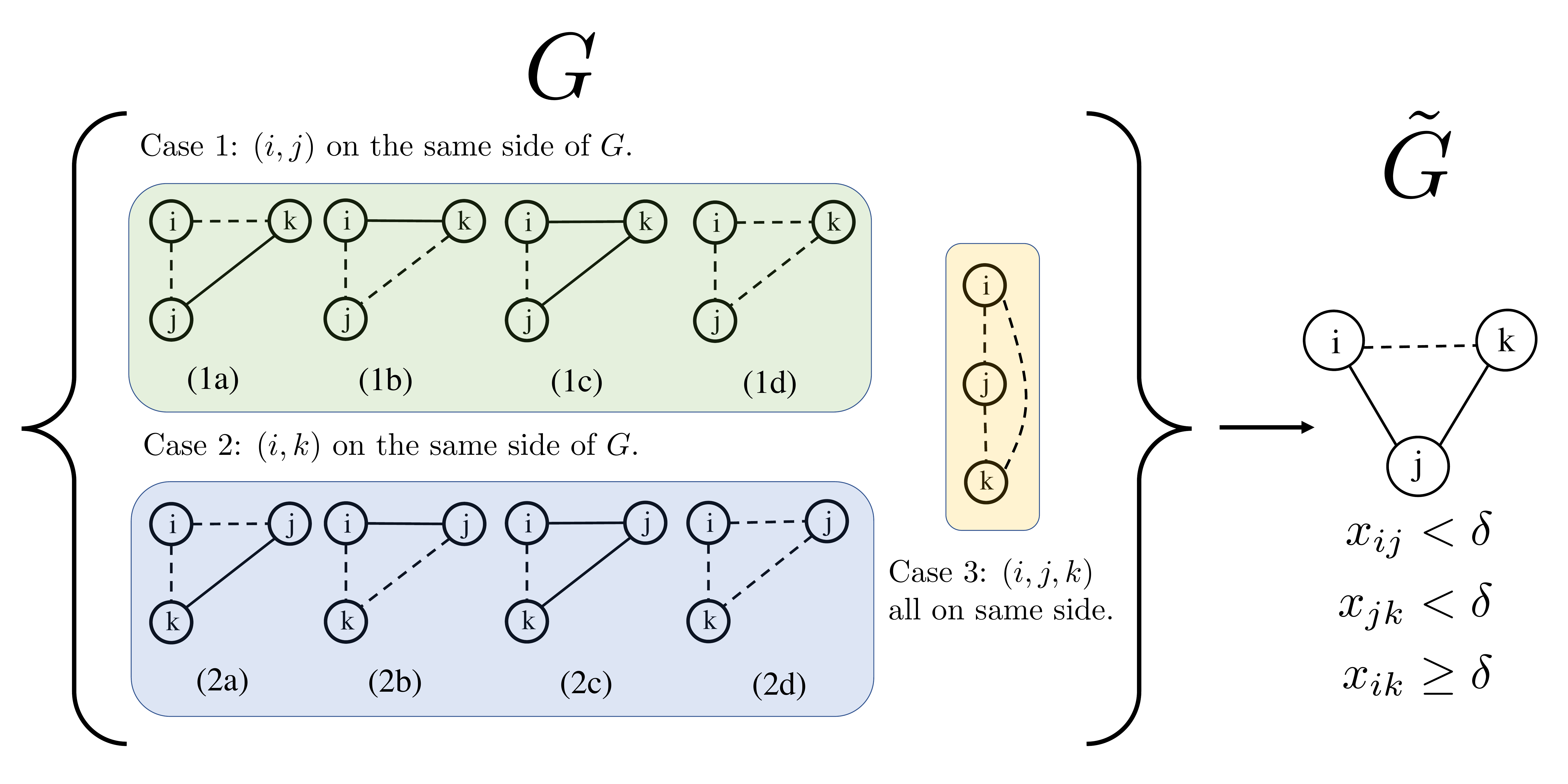}
	\caption{In searching for the best threshold parameter $\delta$, we consider nine different types of triangles that could be mapped to a bad triangle in $\teg$. We will handle inequality~\eqref{eq:triplet} differently depending on the case.}
	\label{fig:trianglecases}
\end{figure}

In Table~\ref{tab:pieces}, we display the value of $L$, and a lower bound on $c_{ij}^{} + c_{jk}^{} + c_{ik}^{}$ for each bad triangle displayed in Figure~\ref{fig:trianglecases}. Let $f_t(\delta)$ denote the lower bound determined for bad triangle of type $t$. 
Full details for computing $L$ and deriving these bounds are presented in the appendix. Once we have such a lower bound $\alpha f_t(\delta) \leq R$ for each type of bad triangle, in order to show that the second condition of Theorem~\ref{thm:3pt1} is always satisfied, it suffices to prove
\begin{equation}
\label{eq:badtri}
L \leq \alpha f_t(\delta)
\end{equation} 
for every bad triangle type $t$.
\begin{table}[t]
	\caption{
		We compute $L = w_{ij}^+ + w_{jk}^+ + w_{ik}^-$ and lower bound on $R/\alpha$ where $R =\alpha(c_{ij}^{} + c_{jk}^{} + c_{ik}^{})$ for each type of bad triangle displayed in Figure~\ref{fig:trianglecases}. Note that the second condition of Theorem~\ref{thm:3pt1} is to check that $L \leq R$ in all cases. For Case 1a, we include two bounds, one that works for any $\beta$, and another bound that is tighter when $\beta \geq 1/2$. 
	}
	\label{tab:pieces}
	\centering
	\renewcommand{\arraystretch}{1.5}
	\begin{tabular}{l l l l}
		\toprule
		\textbf{Case} & $L$ & Bound $ f(\delta) \leq R/\alpha$ \\
		\midrule 
		Case 1a (any $\beta$) & 1& $\mu(1-\delta) +\beta(1-\delta)$\\
		Case 1a ($\beta \geq 1/2$) & 1& $\mu(1-\delta) +\beta + \delta (1-3\beta)$\\
		Case 1b & $1-\beta$& $\mu(1-\delta) + \delta(1-\beta)$\\
		Case 1c & 0 & $ \mu(1-\delta) + \delta (1-\beta)$\\
		Case 1d & $\beta$ &  $\mu(1-\delta) + \beta(2-3\delta)$\\
		Cases 2a, 2b & $(1-\beta) + \mu$ & $\beta(1-\delta) + \mu(1-2\delta)$\\
		Case 2c & $2(1-\beta) + \mu$ & $(1-\beta)\delta+ \mu(1-2\delta)$\\
		Case 2d & $\mu$ & $ 2\beta (1-\delta) +\mu(1-2\delta)$\\
		Case 3 & $\mu$ & $ \mu(3-4\delta)$\\
		\bottomrule
	\end{tabular}
\end{table}

\paragraph{An approximation for $\mu = 0$}
Ailon et al.~\cite{Ailon2011bcc} proved a $4$-approximation for  unweighted bipartite correlation clustering, which is equivalent to PBCC with $\mu =
0$ and  $\beta = 1/2$.
We show how to select~$\delta$ in \textsc{GenRound}
so that not only can we recover this same approximation guarantee when $\mu = 0$ and $\beta = 1/2$, but also obtain guarantees for all $\beta \in \left[ \frac{1}{2}, 1 \right)$. 
\begin{theorem}
	\label{thm:5orbetter}
	When $\mu= \mu_1 = \mu_2 = 0$ and $\beta \geq \frac12$, Algorithm~\ref{alg:genround} with $\delta = \frac{2\beta}{6\beta -1}$ returns a $(6-1/\beta)$-approximation for PBCC.
\end{theorem}
\begin{proof}
	When $\mu = 0$, the system of inequalities  in Table~\ref{tab:pieces} greatly simplifies to the following set of conditions:
	\begin{align}
	1 & \leq \alpha ( 1-\delta) \label{eq:final1} \\
	1 & \leq \alpha [ \beta + \delta(1-3\beta)] \label{eq:final2}\\
	1 & \leq \alpha (2-3\delta)\label{eq:final3} \\
	2& \leq \alpha\delta.  \label{eq:final4}
	\end{align}
	The first of these is a repeat of inequality~\eqref{cond1b}, and the remaining three are derived from Case 1a (the second bound designed specifically for $\beta \geq 1/2$), Case 1d, and Case 2c from Table~\ref{tab:pieces}. One can check to see that all other inequalities we must satisfy are less strict and can be subsumed into one of these four bounds. 
	
	For inequality~\eqref{eq:final1}:
	\begin{align*}
	\alpha (1-\delta) &= (6-1/\beta)\left(1- \frac{2\beta}{6\beta - 1} \right) = \left(\frac{6\beta -1 }{\beta}\right) \left( \frac{4\beta - 1}{ 6\beta - 1}\right)\\
	& = \frac{4\beta -1}{\beta} = 4 - \frac{1}{\beta} \geq 2 > 1.
	\end{align*}
	
	For inequality~\eqref{eq:final2}:
	\begin{align*}
	\alpha [ \beta + \delta(1-3\beta)] &= \frac{6\beta-1}{\beta} \left(\beta + \frac{2\beta}{6\beta - 1}(1-3\beta)  \right) \\
	& = (6\beta - 1) + 2(1-3\beta) = 1.
	\end{align*}
	
	For inequality~\eqref{eq:final3}:
	\begin{align*}
	\alpha (2-3\delta) &= \frac{6\beta-1}{\beta} \left(2- \frac{6\beta}{6\beta - 1} \right) \\
	& = 12 - \frac{2}{\beta} - 6 = 6 - \frac{2}{\beta} \geq 2.
	\end{align*}
	
	For inequality~\eqref{eq:final4}:
	\begin{align*}
	\alpha\delta = \frac{6\beta-1}{\beta} \left( \frac{2\beta}{6\beta - 1} \right) = 2.
	\end{align*}
	All cases are satisfied, and the proof is complete.
\end{proof}

\paragraph{A 5-approx for a generalized parameter regime}
Considering a more general parameter regime, where $\mu = \mu_1 = \mu_2 \in [0,1]$, we obtain a $5$-approximation for all $\beta \geq 1/2$.
\begin{theorem}
	\label{thm:5app}
	When $\mu_1 = \mu_2$ and $\beta \geq \frac12$, Algorithm~\ref{alg:genround}
with $\delta = 2/5$ returns a $5$-approximation to PBCC.
\end{theorem}
\begin{proof}
	In order to prove the result, it is sufficient to show that the following set of inequalities holds when $\delta = 2/5$ and $\alpha = 5$:
	\begin{align}
	1 &\leq 	\alpha\delta \label{eq:first}\\
	1& \leq \alpha(1-\delta) \\
	1 & \leq 	\alpha [\mu(1-\delta) +\beta(1-\delta)] \label{eq:third} \\
	(1-\beta) &\leq \alpha[ \mu(1-\delta) + \delta(1-\beta) ] \label{eq:fourth}\\
	\beta & \leq \alpha [ \mu(1-\delta)+  \beta (2-3\delta) ] \label{eq:fifth}\\
	(1-\beta) + \mu & \leq \alpha [ \beta(1-\delta) + \mu(1-2\delta)] \label{eq:sixth}\\
	2(1-\beta) + \mu & \leq \alpha[(1-\beta)\delta + \mu(1-2\delta)] \label{eq:seventh} \\
	\mu &\leq \alpha [2\beta (1-\delta) +\mu(1-2\delta)] \label{eq:eighth}\\
	1 &\leq  \alpha(3-4\delta).
	\end{align}
	These are taken from~\eqref{cond1a},~\eqref{cond1b}, and all inequalities of the form~\eqref{eq:badtri} obtained from the different cases in Table~\ref{tab:pieces}. The first two inequalities and the last inequality are easy to show just by plugging in $\alpha = 5$ and $\delta = 2/5$.
	For inequalities~\eqref{eq:third},~\eqref{eq:fourth}, and~\eqref{eq:fifth}, we will drop the term $\alpha\mu(1-\delta)$ on the right hand side and prove a more simple set of inequalities that are more strict:
	\begin{align}
	1 & \leq 	\alpha \beta (1-\delta) \label{eq:third2}\\
	1&\leq \alpha\delta \label{eq:fourth2}\\
	1 & \leq \alpha (2-3\delta). \label{eq:fifth2}
	\end{align}
	Inequalities~\eqref{eq:fourth2} and~\eqref{eq:fifth2} follow directly from plugging in $\alpha = 5$ and $\delta = 2/5$. For inequality~\eqref{eq:third2}, note that
	\begin{align*}
	\alpha \beta (1-\delta) = 5 \beta \frac{3}{5} = 3\beta \geq \frac{3}{2} > 1.
	\end{align*}
	
	Finally, note that inequalities~\eqref{eq:sixth},~\eqref{eq:seventh}, and~\eqref{eq:eighth} all have a term $\mu$ on the left and a term $\mu (1-2\delta)$ on the right. So all of these inequalities will be satisfied if we prove the more strict conditions 
	\begin{align}
	\mu &\leq \alpha\mu (1-2\delta) \label{eq:mupiece}\\
	(1-\beta) & \leq \alpha \beta(1-\delta) \label{eq:sixth2}\\
	2(1-\beta)& \leq \alpha (1-\beta)\delta \label{eq:seventh2}
	\end{align}
	Note that~\eqref{eq:mupiece} holds tightly for $\alpha = 5$ and $\delta = 2/5$. 
	Inequality~\eqref{eq:sixth2} is less strict that inequality~\eqref{eq:fourth2}, which we already proved.
	Finally, after canceling $(1-\beta)$ from both sides,~\eqref{eq:seventh2} becomes $2 \leq \alpha \delta$, which holds for our choices of $\alpha$ and $\delta$. 
	Thus, all necessary constraints are satisfied and we know that Algorithm~\ref{alg:genround} will yield a 5-approximate solution for any $\mu$ and whenever $\beta \geq 1/2$.
\end{proof}

\subsection{Modularity Connections and Heuristics} 
Returning to the \textsc{HyperLam} objective, applying our weighted clique expansion and introducing a negative edge of weight $\lambda d_i d_j$ for
node pair $(i,j)$ is equivalent to solving a weighted variant of the LambdaCC graph clustering objective~\cite{Veldt:2018:CCF:3178876.3186110}. Since
LambdaCC is equivalent to a generalization of modularity with a resolution parameter~\cite{Veldt:2018:CCF:3178876.3186110,newman2004modularity}, we
can also approximately optimize the \textsc{HyperLam} objective by applying our weighted clique expansion and then running heuristic algorithms for
modularity such as the Louvain algorithm~\cite{Blondel_2008} or, more appropriately, generalizations of Louvain with a resolution parameter~\cite{JeubGenLouvain}. A similar approach will also work for the star expansion: we set the weight of a node in $V$ to be its hyperedge degree $w_v = d_v$, and the weight of an auxiliary node $v_e$ (obtained from expanding a hyperedge) to be $w_{v_e} = 0$. This also corresponds to a weighted variant of LambdaCC, since each pair of nodes $(i,j)$ in the graph share a negative edge of weight $\lambda w_i w_j$. In many cases this weight will be zero, but we can still apply generalized Louvain-style heuristics to optimize the objective.


Kumar et al.~\cite{kumar2020new} previously considered a modularity-based approach for hypergraph clustering based on the same type of clique expansion. These authors applied the same weight $1/(|e|-1)$ to each edge in a clique expansion of a hypergraph $|e|$, as this preserves the degree distribution of nodes in the original hypergraph.
They then considered applying the modularity objective~\cite{newman2004modularity} to the resulting graph. Their approach corresponds to applying a weighted clique expansion to an instance of \textsc{HyperLam}, and setting $\lambda = 1/(\vol_\mathcal{H}(V))$. Thus, this approach can be viewed as a special case of our hyperedge expansion procedure for \textsc{HyperLam}. The connection to correlation clustering we show, along with the resulting approximation algorithms for the all-or-nothing hypergraph cut, provide further theoretical motivation for this choice of weighted clique expansion. Despite this connection to a previous clique expansion technique for modularity, we note that our original hypergraph objective~\eqref{eq:hocc} nevertheless differs from generalizations of modularity defined directly for hypergraphs~\cite{kaminski2019clustering}, as opposed to modularity objectives applied to clique expansions of hypergraphs.

\section{Related Work}
To anchor our work, we highlight related results on algorithms for correlation clustering, techniques for parametric clustering in standard graphs, and recent results on clustering hypergraphs. 

\textbf{Correlation Clustering} Bansal et al.~\cite{BansalBlumChawla2004} first introduced the problem of correlation clustering, providing a constant factor approximation for the complete unweighted case. 
Amit was the first to consider the problem in the bipartite setting~\cite{amit2004bicluster}, providing an $11$-approximation for the complete
unweighted setting. Later, Ailon et al.~\cite{Ailon2011bcc} presented a $4$-approximation. Most recently, Chawla et
al.~\cite{ChawlaMakarychevSchrammEtAl2015} improved the best approximation factor to~$3$.

Higher-order correlation clustering was first considered by Kim et al.~\cite{kim2011highcc} in the content of image segmentation. Li et
al.~\cite{Li2017motifcc} were the first to develop approximation algorithms for the complete 3-uniform case, giving a 9-approximation. We later gave a $4(k-1)$ approximation for the $k$-uniform setting, which was then improved to $2k$ by Li et al.~\cite{Li2019motif}. For weighted hypergraphs, Fukunaga~\cite{fukunga2018highcc} presented an $O(k \log n)$ approximation algorithm, where $k$ is the maximum size of negative hyperedges.

\textbf{Parametric Graph Clustering}
Our introduction of the LambdaCC framework situates graph clustering within correlation clustering~\cite{Veldt:2018:CCF:3178876.3186110}.
We proved equivalence results with modularity, normalized cut, and sparsest cut, and gave a $3$-approximation when $\lambda \geq 1/2$, based on LP-rounding. We were later able to show that the LP relaxation has an integrality gap of $O(\log n)$ for some small values of $\lambda$~\cite{Veldt2018ccgen}. LambdaCC is in turn related to other graph parametric clustering objectives, such as stability~\cite{Delvenne2010stabilitypnas}, various Potts models~\cite{ReichardtBornholdt2004,traag2011narrow}, and generalizations of modularity~\cite{Arenas2008analysis}. 

\textbf{Hypergraph Clustering}
Different higher-order generalizations of modularity have been previously
developed~\cite{kumar2020new,kaminski2019clustering}, along with higher-order variants of
conductance~\cite{BensonGleichLeskovec2016} and normalized cut~\cite{panli2017inhomogeneous,Zhou2006learning}. In hypergraph clustering, the most
common penalty for a cut hyperedge is the weight of that hyperedge, regardless of how the hyperedge is cut. However, other penalties have also been
considered in the context of hypergraph partitioning and clustering~\cite{Catalyurek99hypergraph,panli_submodular, panli2017inhomogeneous}. A more comprehensive overview of generalized hypergraph cut functions is included in recent work by one of the authors~\cite{veldt2020hypergraph}.

\section{Experiments}
\label{sec:experiments}
We demonstrate our parametric objectives and algorithms in analyzing an assortment of different types of datasets. Our primary goal is to highlight the diversity of results we can achieve. We begin by running our approximation algorithms for PBCC on several bipartite datasets to illustrate the algorithmic performance and output in different parameter regimes. We then apply the \textsc{HyperLam} framework to motif clustering. Finally, we apply our framework to detect product categories in an Amazon product review hypergraph.


\textbf{Implementation Details.}
We implement our algorithms in Julia, using Gurobi to solve LP relaxations. Code for all algorithms and experiments are
available online at~\url{https://github.com/nveldt/ParamCC}. We focus on studying the differences among the objective
functions rather than optimizing implementations. Our motif clustering experiments were run on a laptop with 8GB of RAM.
All other experiments were run on a larger machine with four~16-core Intel Xeon E7-8867 v3 processors.
Running large instances with Louvain-style algorithms was not a bottleneck and these always finished in a few minutes or
less. On the bipartite graphs we consider, running our PBCC algorithms typically took a few seconds or a few minutes.
Solving the correlation clustering LP relaxation for larger graphs is often very expensive; this is, however, an active research area~\cite{ruggles2020full,veldt2019simods,brickell2008metricnearness,sonthalia2020project} 
and solvers have been produced for around~20,000-node graphs. This leaves us with a theory/practice gap between the effective Louvain-based heuristics and more principled approximations that we intend to study in the future.

\subsection{PBCC on Real Bipartite Graphs}
We run our PBCC approximation algorithms on five bipartite graphs constructed from real data\footnote{Cities: \url{https://www.lboro.ac.uk/gawc/datasets/da6.html}; Newsgroups:~\url{www.cs.nyu.edu/~roweis/data/}; Zoo: \url{https://archive.ics.uci.edu/ml/datasets/zoo}. Amazon (5-core):~\url{https://nijianmo.github.io/amazon/index.html}.}, with a range of parameter settings. 
\setlength{\leftmargini}{0pt} 
\begin{compactitem}
\item The \emph{Cities} graph encodes which set of 46 global firms (nodes on side $V_1$) have offices in 55 different major cities (nodes on side $V_2$).  
 \item\emph{Newgroups100} is made up of a set of 100 documents ($V_1$) and 100 words ($V_2$); edges indicate words used in each document. We have extracted a random subset of 100 documents (25 from each of four categories: $\textit{sci}*$, $\textit{comp}*$, $\textit{rec}*$, and $\textit{talk}*$) from a larger dataset,
often used as a benchmark for \emph{hypergraph} clustering~\cite{Zhou2006learning,panli_submodular,Hein2013}. 
\item The \emph{Zoo} dataset encodes 100 animals and their associations with~15 different binary attributes (e.g., ``hair'', ``feathers'', ``eggs''). 
\item The last two bipartite graphs are constructed from reviewers on Amazon ($V_1$) that have reviewed products ($V_2$) within  certain categories~\cite{ni-etal-2019-justifying}. The \emph{Fashion} category has 404 reviewers and 31 products, and \emph{Appliances} has 44 reviewers for 48 products.
\end{compactitem}
 \begin{figure}[t]
	\centering
	\subfloat[$\mu =0$, $0 \leq \beta \leq 1$ \label{fig:beta}] 
	{\includegraphics[width=.47\linewidth]{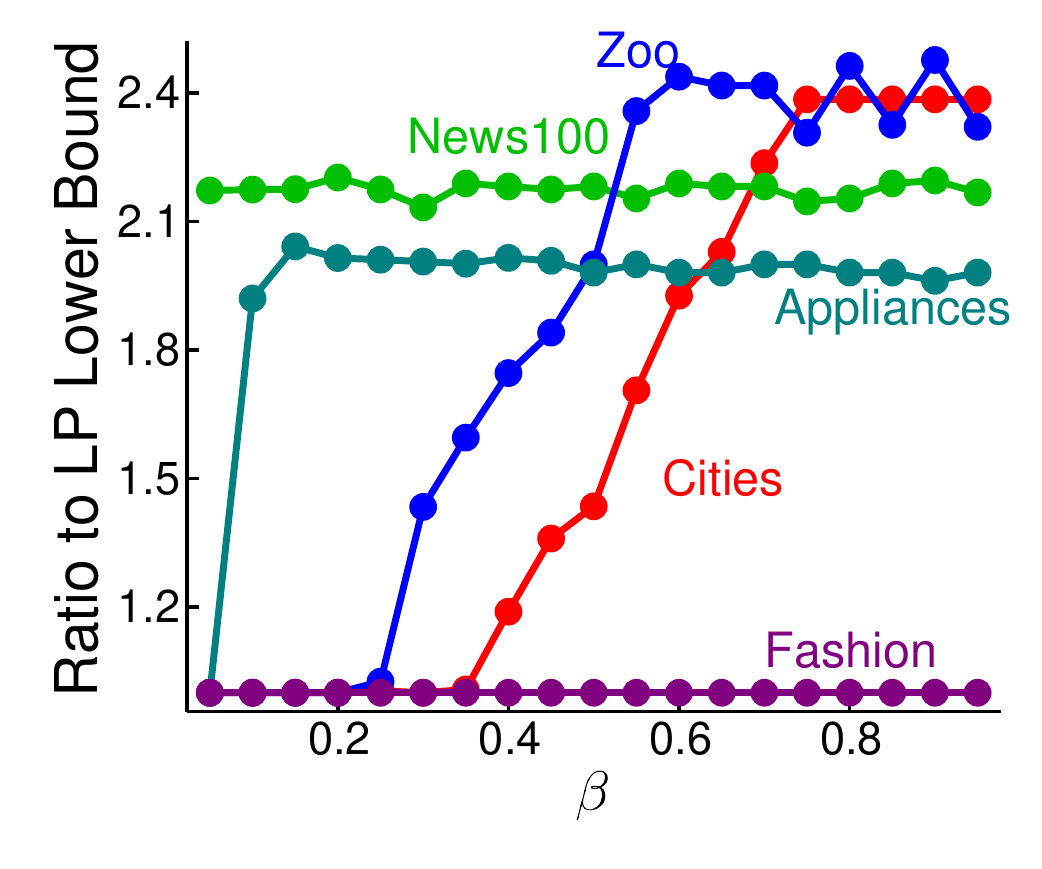}}\hfill
	\subfloat[$\beta = 0.5$, $0 \leq \mu\leq 0.2$ \label{fig:mu}] 
	{\includegraphics[width=.47\linewidth]{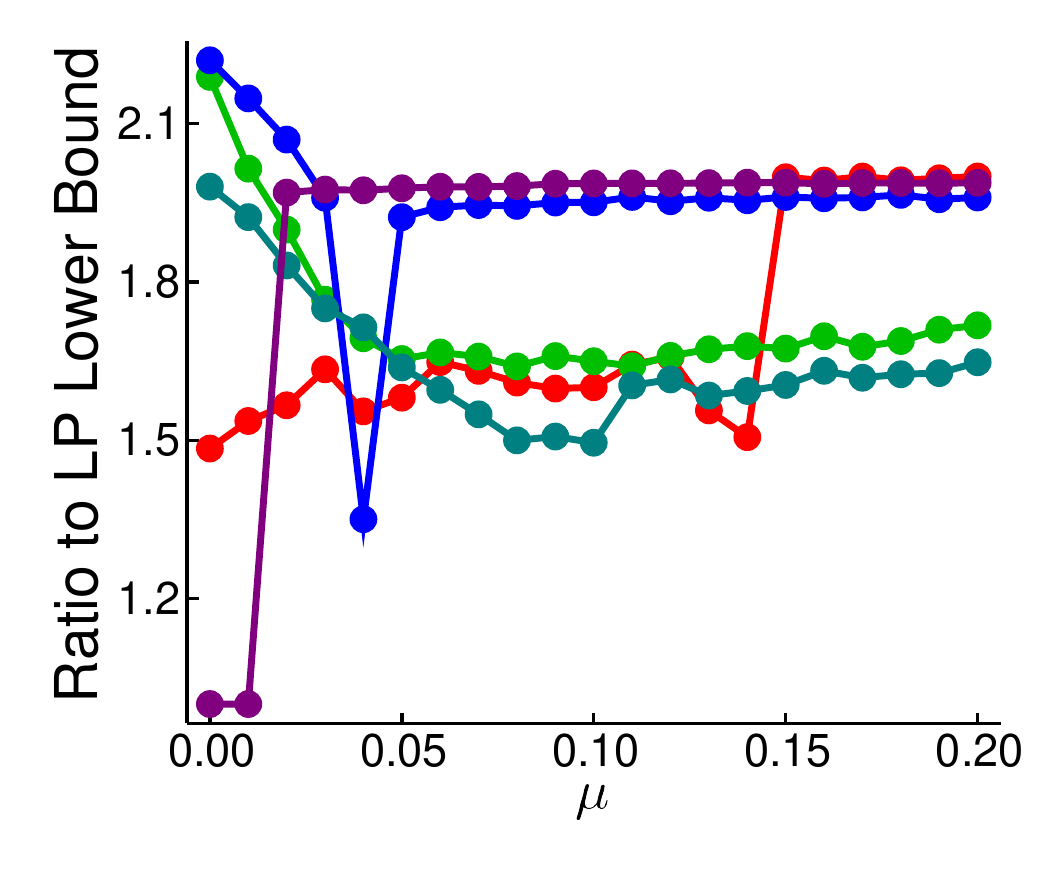}}
		\vspace{-.5\baselineskip}
	\caption{A posteriori approximation ratios for running our LP-based PBCC algorithms on  real-world bipartite graphs.
	}
	\label{fig:pbcc-plots}
	\vspace{-\baselineskip}
\end{figure}
Figure~\ref{fig:pbcc-plots} displays a posteriori approximation ratios for our method (objective score divided by LP lower bound), first for $\mu_1 = \mu_2 = 0$ and $\beta \in [0,1]$, and then for $\beta = 1/2$ and $\mu = \mu_1 = \mu_2 \in [0, 0.2]$. After solving the LP relaxation for each $(\mu, \beta)$ pair, we try rounding with~$\delta$ values from $0.05$ to $0.95$ in increments of $0.05$, taking the result with the best objective score, since the rounding procedure is much faster than the initial LP solve. We note that the approximation factor curve varies significantly from dataset to dataset. However, in all cases we obtain much better approximation factors than the ones given in Table~\ref{tab:pbcc}, even for $\beta$ values where our algorithms have no formal guarantees. In certain regimes we also observe abrupt changes in approximation factors, e.g., for \emph{Fashion} when $\beta = 0.5$ and $\mu$ is near zero (Figure~\ref{fig:mu}). We also tested $\mu > 0.2$ when $\beta = 0.5$. In this parameter regime, the problem is nearly the same as bipartite matching, though our LP-based approach only provides a posteriori guarantees of around a factor 2. This motivates the question of what other approximation algorithms might perform better when the problem is ``almost'' bipartite matching.

\subsection{\textsc{HyperLam} for Motif Clustering}
\textsc{HyperLam} can detect motif-rich clusters at different resolutions in a graph. In motif clustering, a small, frequently repeated subgraph (a motif) is identified, and each motif instance is associated with a hyperedge~\cite{BensonGleichLeskovec2016,Arenas_2008,Tsourakakis:2017:SMG:3038912.3052653,panli2017inhomogeneous}. Applying a hypergraph clustering technique penalizes the number of cut motifs, rather than just the cut edges. This encourages keeping whole motifs inside clusters. 

Triangles are known to be important motifs for identifying community structure in networks~\cite{Tsourakakis:2017:SMG:3038912.3052653,klymko2014using}. We therefore apply the \textsc{HyperLam} framework to cluster the Email-EU dataset~\cite{yin2017local,leskovec2007graph} based on triangles. Each edge in the graph (which we treat as undirected) represents an email sent between members of a European research institution. A  metadata label indicating each researcher's department comes with each node. 

To find clusters at different resolutions in the graph, we approximate the \textsc{HyperLam} objective by first applying
a clique expansion based on triangle motifs. Since the motif has three nodes, the all-or-nothing cut is the same as the
linear penalty, and the clique expansion perfectly models both. We cluster the resulting \emph{weighted} graph with a
weighted version of Lambda-Louvain~\cite{Veldt:2018:CCF:3178876.3186110}, which makes greedy local node moves similar to
the Louvain method~\cite{Blondel_2008}, but optimizes a different objective. We compare against running Lambda-Louvain
on the original graph. We also compare against standard graph algorithms Metis~\cite{karypis1998metis} and
Graclus~\cite{Dhillon-2007-graclus}, varying the number of clusters $k$, and recursive spectral partitioning, for a
range of different minimum cluster sizes $m_\mathit{size}$. We test these last three methods on both the original graph
and clique-expanded graph, but show results only for the clique-expanded graph, as this leads to the best outcome for
these methods. Finally, we run hMetis (a hypergraph variant of Metis) on the hypergraph formed by associating motifs
with hyperedges, varying cluster number,~$k$.

After forming multiple clusterings with each method for many parameter values ($k$, $m_\mathit{size}$, or $\lambda$), we measure the Adjusted Rand Index score between each clustering and the known department metadata labels. Scores for each cluster size are displayed in Figure~\ref{fig:email}. Although the department labels do not exactly match with community structure in the network, there is a strong correlation between the two, and the higher ARI scores obtained by running \textsc{HyperLam} with the triangle motif indicate that our method is best able to detect this relationship. 
	
	\begin{figure}[t]
		\centering
		\subfloat[Email, ARI scores \label{fig:email}] 
		{\includegraphics[width=.45\linewidth]{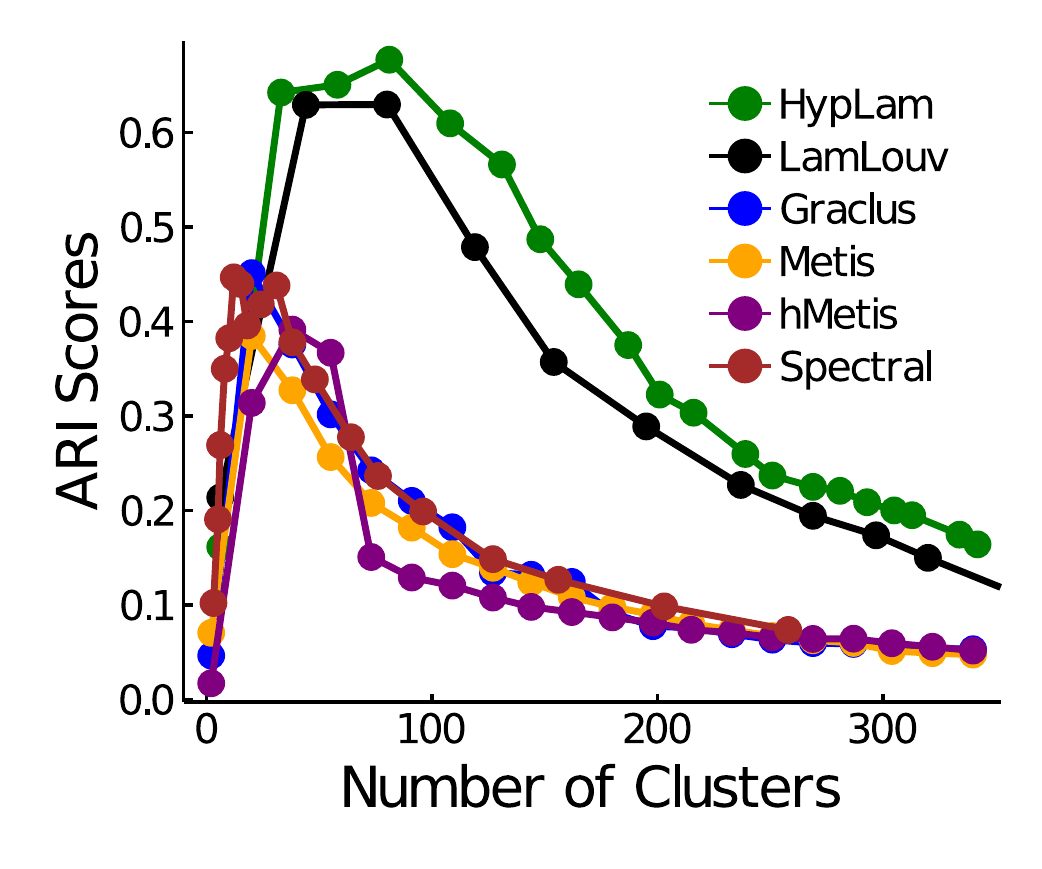}}\hfill
		\subfloat[Florida Bay, ARI Scores \label{fig:florida}] 
		{\includegraphics[width=.45\linewidth]{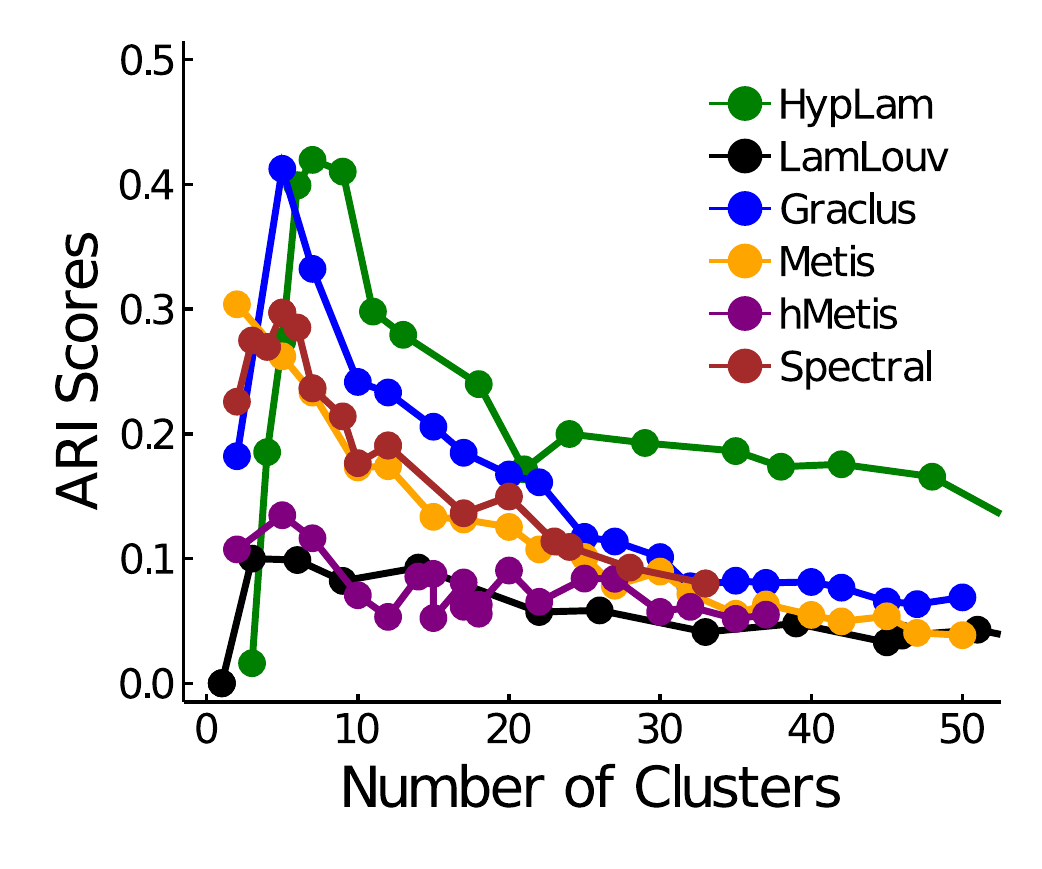}}\hfill
		\vspace{-.5\baselineskip}
		\subfloat[Email, runtime (seconds) \label{fig:emairunl}]
		{\includegraphics[width=.45\linewidth]{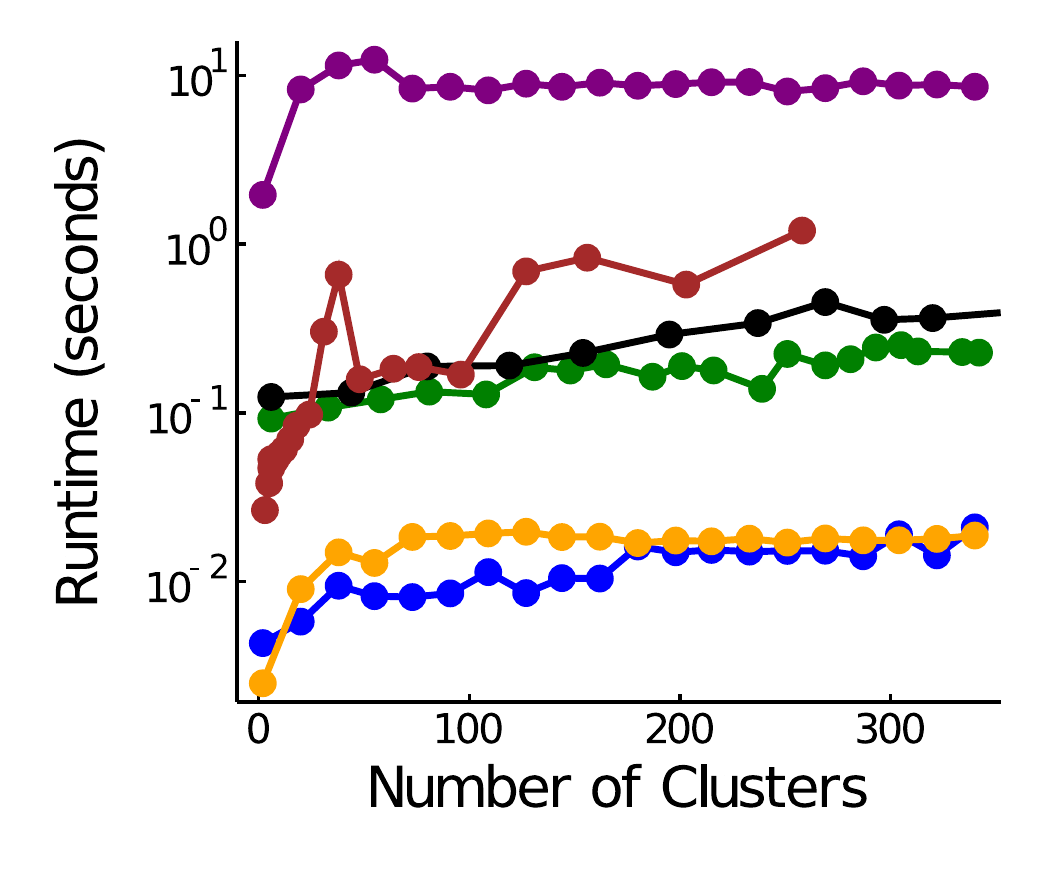}}\hfill
		\subfloat[Florida Bay, runtime (seconds) \label{fig:floridarun}] 
		{\includegraphics[width=.45\linewidth]{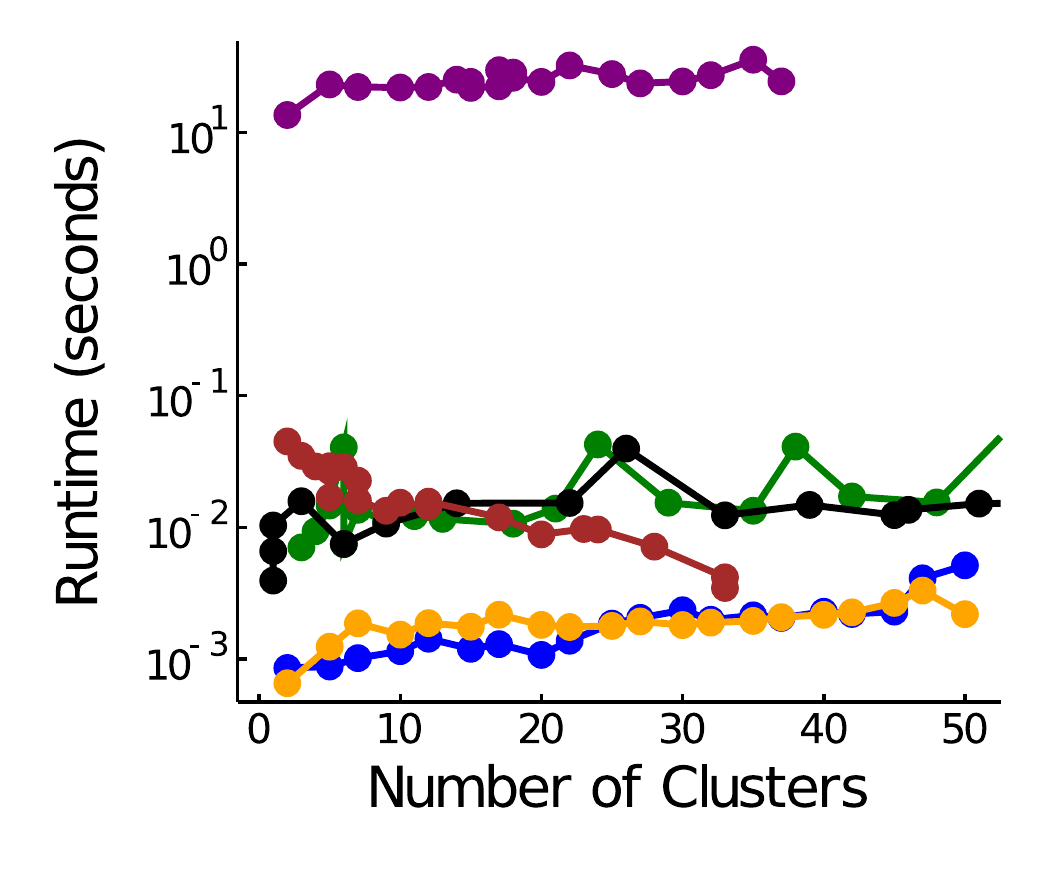}}\hfill
		\vspace{-.5\baselineskip}
		\subfloat[Email, runtime/ARI \label{fig:emailtrade}]
		{\includegraphics[width=.45\linewidth]{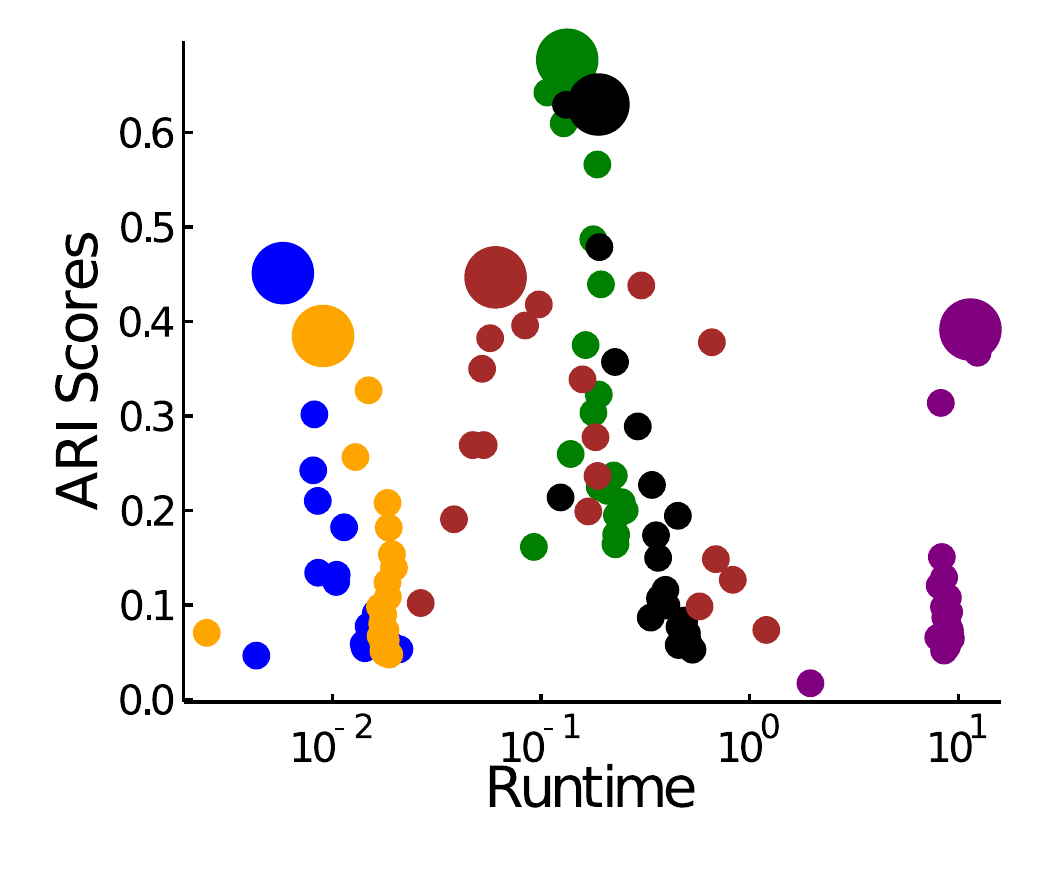}}\hfill
		\subfloat[Florida Bay, runtime/ARI\label{fig:floridatrade}]
		{\includegraphics[width=.45\linewidth]{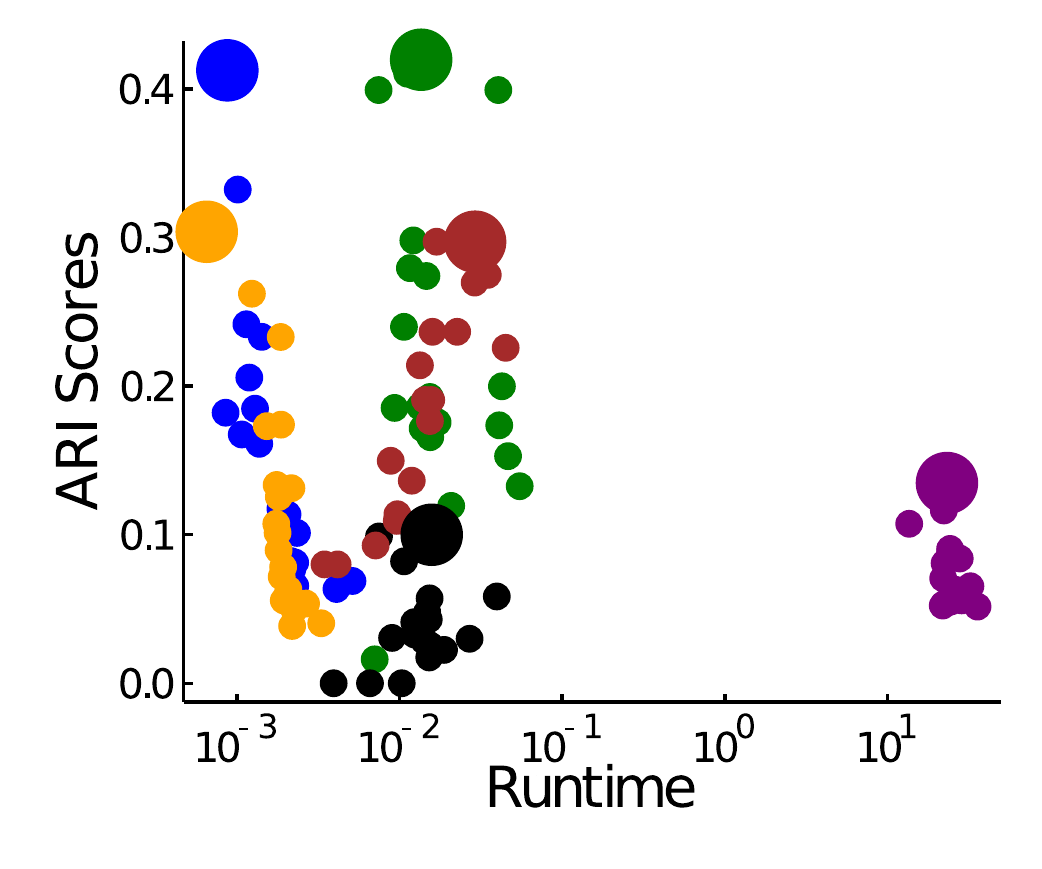}}\hfill
		\vspace{-.5\baselineskip}
		\caption{(a) \textsc{HyperLam} with triangle motifs better captures the relationship between community structure and department labels of researchers at a European research institution, across all clusters sizes. (b) Optimizing \textsc{HyperLam} using the \emph{bifan} motif and the inhomogeneous hyperedge splitting function of Li et al.~\cite{panli2017inhomogeneous}, we find clusterings with higher correlation with biological classifications of species in a food web. Figures (c) and (d) display runtimes, while (e) and (f) display the
trade-off between runtime and ARI score. Larger dots mark the best ARI score for each method.}
		\label{fig:motif-plots}
		\vspace{-\baselineskip}
	\end{figure}

We perform a similar experiment on the Florida Bay food web, in which nodes indicate species (e.g., Isopods, Eels, Meroplankton), and directed edges indicating carbon exchange~\cite{panli2017inhomogeneous,BensonGleichLeskovec2016}. Following the approach of Li and Milenkovic~\cite{panli2017inhomogeneous}, we consider the \emph{bifan} motif, in which two nodes $\{v_1, v_2\}$ have uni-directional edges to two other nodes $\{v_3, v_4\}$, and any edge combination within sets $\{v_1, v_2\}$ and $\{v_3, v_4\}$ is allowed. We identify each instance of the motif as a hyperedge. Li and Milenkovic specifically use an \emph{inhomogeneous} hyperedge cutting penalty, which can be modeled by simply adding undirected edges $(v_1, v_2)$ and $(v_3,v_4)$. Thus, we convert the input graph into a new graph, and cluster with a weighted version of Lambda-Louvain, to optimize the \textsc{HyperLam} objective. We again run hMetis on the hypergraph defined by motifs, and Lambda-Louvain on the undirected version of the original graph. We ran $\{\text{Metis},\text{Graclus}, \text{Recursive spectral} \}$ on the new graph obtained by expanding bifan motifs, as this led to better results than running them on the original graph. Figure~\ref{fig:florida} demonstrates that applying our \textsc{HyperLam} framework with the bifan motif structure leads to the highest ARI clustering scores with the biological classifications identified by Li et al.~\cite{panli2017inhomogeneous} (e.g.\ producers, fish, mammals). Acknowledging that our implementations are not optimized for speed,
Figures~\ref{fig:emailtrade} and~\ref{fig:floridatrade}
show that Metis, Graclus, and \textsc{HyperLam} methods constitute the efficient
frontier.

\subsection{Clustering Amazon Products Categories }
In our last experiment we illustrate differences that arise when applying the \textsc{HyperLam} framework with different hyperedge cut functions. In order to do
so, we apply our framework to a hypergraph constructed from Amazon review data, similar to the \emph{Fashion} and ~\emph{Appliances} hypergraphs in
the first experiment. This time, we extract nine product categories, associating each product in these categories with a node, and defining a
hyperedge to be a set of all products that are reviewed by the same person. This
results in a hypergraph with~13,156 nodes,~31,544 hyperedges, with the maximum and mean hyperedge sizes being~219 and~8.1, respectively. Each node is associated with exactly one category label.

As outlined in Section~\ref{sec:algs}, we apply a weighted clique expansion and a star expansion to the Amazon review hypergraph, each modeling a different cut penalty.
We scale the graphs so that they share the same total volume, then cluster them both with Lambda-Louvain, using various values of~$\lambda$. Running Lambda-Louvain on the clique expansion took just over two minutes on average, while runtimes were just over four minutes on average for the star expansion. 

The hypergraph has a single large connected component, indicating that reviewers do review products across different categories. At the same time,~95\% of all hyperedges in the hypergraph are completely contained inside one of the sets of nodes defining a product category. Thus, we expect that clustering the hypergraph based on hyperedge structure will yield clusters that correlate highly with product categories. We confirm this by computing ARI scores between category labels and the clusterings returned by optimizing \textsc{HyperLam} for both graph expansions (Figure~\ref{fig:amazon}). 

In order to better understand the structure of clusters formed by our methods, and their relationship with product categories, we measure how well each clustering detects individual product-category node sets in the hypergraph. For each category (e.g., ``Appliances''), we measure how well a \textsc{HyperLam} clustering ``tracks'' that category by taking the best~F1 score between any of the \textsc{HyperLam} clusters and the product-category node set in question. For example, if one of the clusters returned by \textsc{HyperLam} exactly matches the ``Appliances'' node set, then we have perfectly ``tracked'' this category, and we report an F1 score of 1.
Figure~\ref{fig:clusters} illustrates that in general, the star expansion is able to better track the two largest categories, ``Prime Pantry'' and
``Industrial \& Scientific'', each of which has roughly~5000 nodes. This helps explain why the star expansion obtains higher ARI scores in general.
On the other hand, we observed that the clique expansion tracks the ``Software'' category (802 nodes) better. This highlights the fact that different
hyperedge cut functions can leads to substantially different types of clusters. 
 \begin{figure}[t]
	\centering
	\subfloat[ARI scores  \label{fig:aris}]
	{\includegraphics[width=.47\linewidth]{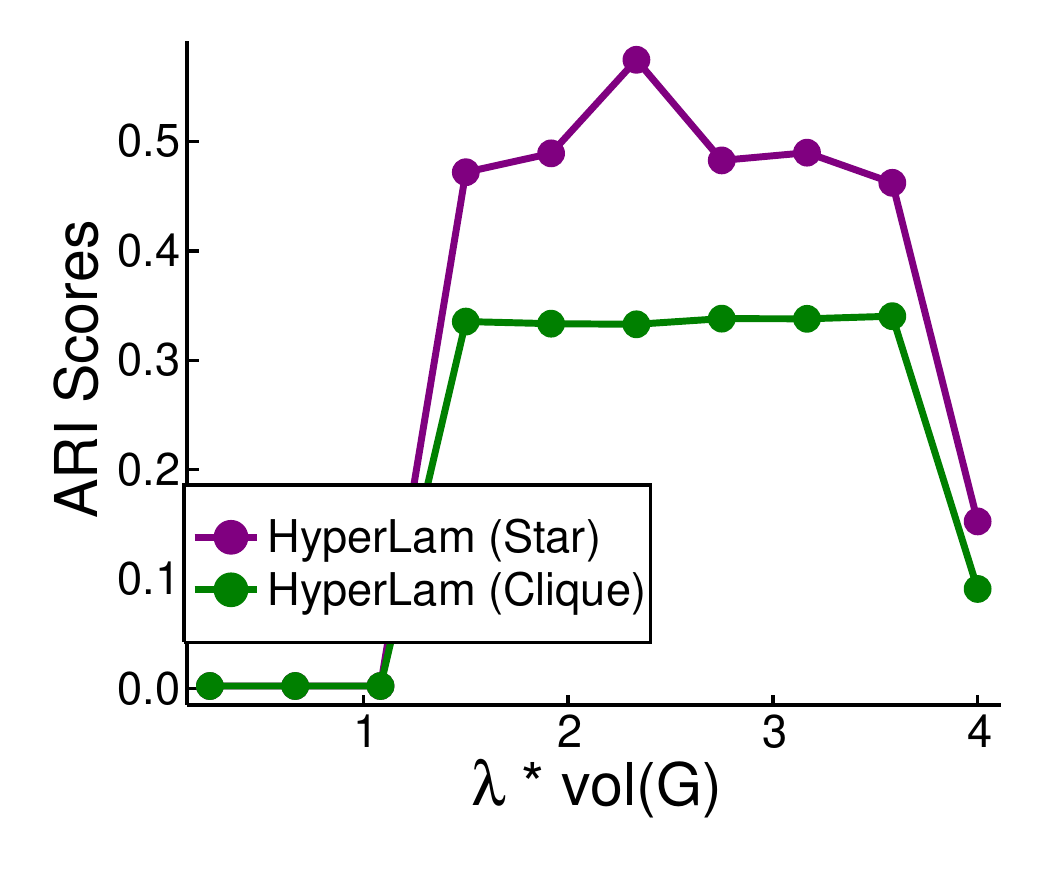}}\hfill
	\subfloat[Tracking categories \label{fig:clusters}] 
	{\includegraphics[width=.47\linewidth]{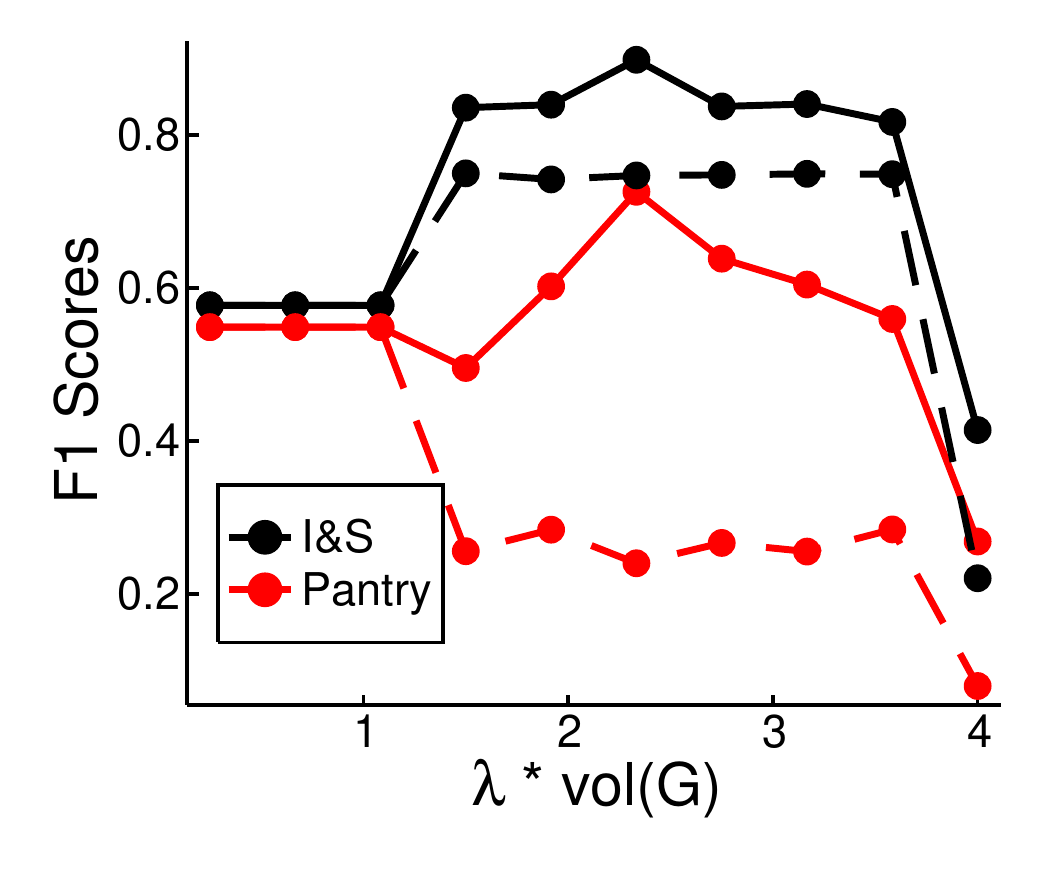}}
	\vspace{-.5\baselineskip}
	\caption{(a) The clique and star expansion lead to clusterings that are correlated with product categories in an Amazon product hypergraph. (b) We compute the best F1 score between clusters formed by \textsc{HyperLam}, and individual product category clusters. The star expansion (results with solid lines) is able to better track the two largest clusters, ``Industrial and Scientific" (black) and ``Prime Pantry'' (red), compared to the clique expansion (dashed lines). 
	}
	\label{fig:amazon}
	\vspace{-\baselineskip}
\end{figure}

\section{Discussion}
We have presented a new, flexible, and general framework for parametric clustering of hypergraph and bipartite graph datasets. This framework has deep connections to existing objective functions in the literature and there exist polynomial time approximation results as well as heuristic algorithms. While such frameworks are extremely useful to expert practitioners to engineer and investigate datasets, they are often challenging for less sophisticated users who have a tendency to rely on default parameters. Towards that end, there is a general need for statistical and automated techniques to help guide users to the most successful use of these methods, which is something we hope to design in the future.

Another challenge with the methods involves scaling of the parameters. In our experiments, we often scale these by the volume of the graph (the total sum of edge-weighted degrees) as that has proven to be successful in practice. However, it is unclear if this is the best approach in all circumstances, or whether there are situations in which absolute values of the parameters should be preferred. Finally, as our experiments highlight, there are distinct phase transitions in the behavior among these different regimes; finding ways to identify these characteristic regions would also make these parametric objectives useful to automatically find characteristically different clusterings.

\begin{acks}
This research was supported by NSF IIS-1546488, CCF-1909528, NSF Center for Science of Information STC, CCF-0939370, DOE DESC0014543, NASA, the Sloan Foundation, and the Melbourne School of Engineering.
\end{acks}

\bibliographystyle{plain}
\bibliography{param-cc}
\appendix
\section{Proofs for Bounds in Table~\ref{tab:pieces}.}
\label{sec:appsystem}
In order to prove approximation guarantees for PBCC, we must determine the best way to satisfy the first and second condition in Theorem~\ref{thm:3pt1}. Details for satisfying the first condition are given in the main text. For the second condition, we consider each triangle from Figure~\ref{fig:trianglecases} in turn, each with their own accompanying figure.
In each case, we state the left hand side $L =  w_{ij}^+ + w_{jk}^+ + w_{ik}^- $ in terms of $\mu$ and $\beta$, and then bound $c_{ij}^{} + c_{jk}^{} + c_{ik}^{}$ below by some linear function $f(\delta)$. We know then that for each case we must satisfy
\[ \frac{f(\delta)}{L} \geq \frac{1}{\alpha}\,. \]
Table~\ref{tab:pieces} in the main text summarizes the bounds we compute here. Recall throughout that $x_{ij}^{} < \delta$, $x_{jk}^{} < \delta$,  and $x_{ik}^{} \geq \delta$.


\begin{figure}[h]
	\begin{minipage}[c]{0.2\linewidth}
		\centering
		{\includegraphics[width=\linewidth]{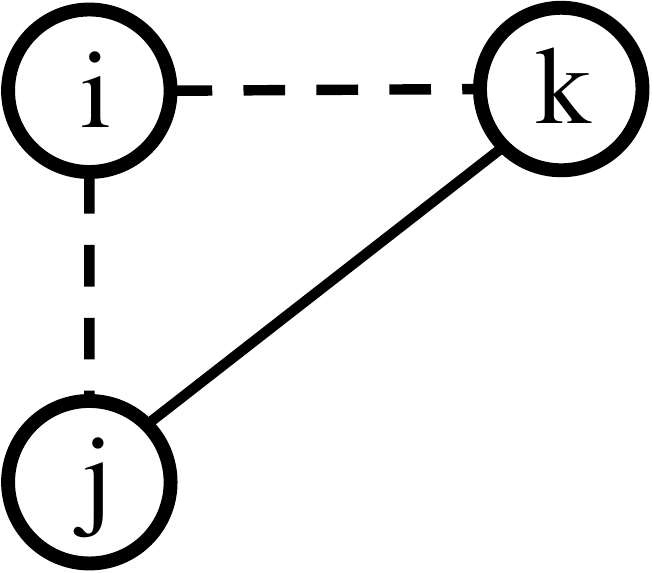}}
		\caption*{Case (1a)}\label{1a}
	\end{minipage}
	\hfill 
	\begin{minipage}[c]{0.75\linewidth}
		$L =  w_{ij}^+ + w_{jk}^+ + w_{ik}^- = (1-\beta) + \beta = 1$.\newline
		
		If $\beta \geq 1/2$:
		\begin{align*}	
		& c_{ij}^{} + c_{jk}^{} + c_{ik}^{} \\
		&= \mu(1-x_{ij}) + (1-\beta)x_{jk}^{} + \beta (1-x_{ik}) \\
		&= \mu(1-x_{ij}) + (2\beta -1)(1-x_{ik}) \\& \hspace{.5cm} + (1-\beta)(x_{jk}^{} + 1- x_{ik}) \\
		& \geq \mu(1-\delta) + (2\beta-1)(1-2\delta) + (1-\beta)(1-x_{ik})\\
		& \geq \mu(1-\delta) + (2\beta-1)(1-2\delta) + (1-\beta)(1-\delta)\\
		&= \mu(1-\delta) +\beta + \delta (1-3\beta).
		\end{align*}
		For any $\beta$: $	c_{ij}^{} + c_{jk}^{} + c_{ik}\geq \mu(1-\delta) + \beta (1-\delta)$.
		
	\end{minipage}
	
\end{figure}

\begin{figure}[h]
	\begin{minipage}[c]{0.2\linewidth}
		\centering
		{\includegraphics[width=\linewidth]{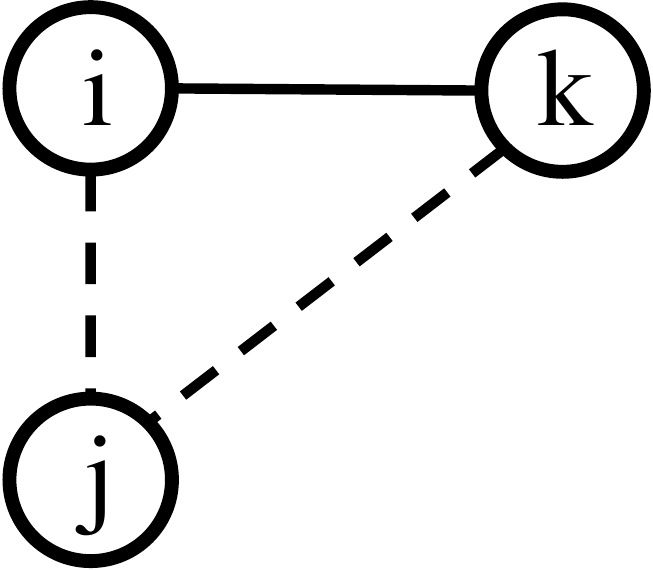}}
		\caption*{Case (1b)}	\label{1b}
	\end{minipage}
	\hfill 
	\begin{minipage}[c]{0.5\linewidth}
		
		$L =  w_{ij}^+ + w_{jk}^+ + w_{ik}^- = (1-\beta)$.
		
		\begin{align*}	
		c_{ij}^{} + c_{jk}^{} + c_{ik}^{} &= \mu(1-x_{ij}) + (1-\beta)(x_{jk}+x_{ik}) \\
		&\geq \mu(1-\delta) + (1-\beta)\delta.\\
		\end{align*}
		
	\end{minipage}
	
\end{figure}

\begin{figure}[h!]
	\begin{minipage}[c]{0.2\linewidth}
		\centering
		{\includegraphics[width=\linewidth]{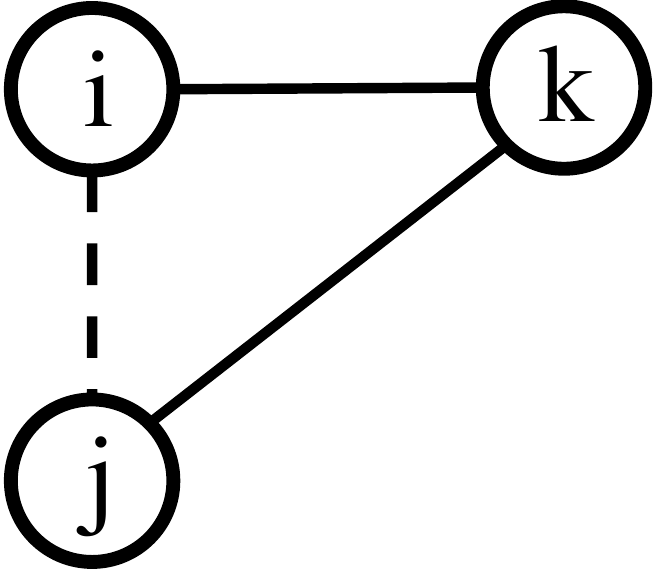}}
		\caption*{Case (1c)}	\label{1c}
	\end{minipage}
	\hfill
	\begin{minipage}[c]{0.7\linewidth}
		
		$L =  w_{ij}^+ + w_{jk}^+ + w_{ik}^- = 0$.\newline
		
		The inequality is trivial since left hand side is zero.
		
	\end{minipage}
	
\end{figure}

\begin{figure}[h!]
	\begin{minipage}[c]{0.2\linewidth}
		\centering
		{\includegraphics[width=\linewidth]{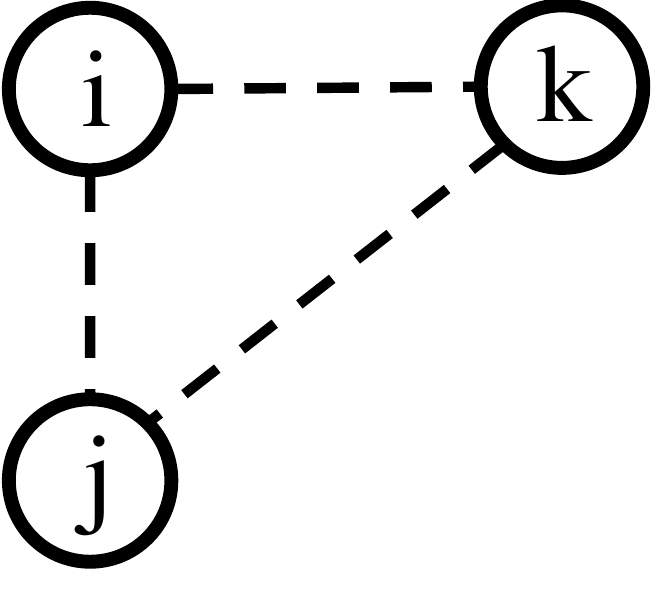}}
		\caption*{Case (1d)}	\label{1d}
	\end{minipage}
	\hfill 
	\begin{minipage}[c]{0.7\linewidth}
		
		$L =  w_{ij}^+ + w_{jk}^+ + w_{ik}^- = \beta$.
		
		\begin{align*}	
		c_{ij}^{} + c_{jk}^{} + c_{ik}^{} &=  \mu(1-x_{ij}) + \beta( 1- x_{jk}^{} + 1- x_{ik}) \\
		&\geq \mu(1-\delta) + \beta( 1- \delta + 1- 2\delta)\\
		&= \mu(1-\delta) + \beta( 2- 3\delta).
		\end{align*}
		
	\end{minipage}
	
\end{figure}

\begin{figure}[h!]
	\begin{minipage}[c]{0.2\linewidth}
		\centering
		{\includegraphics[width=\linewidth]{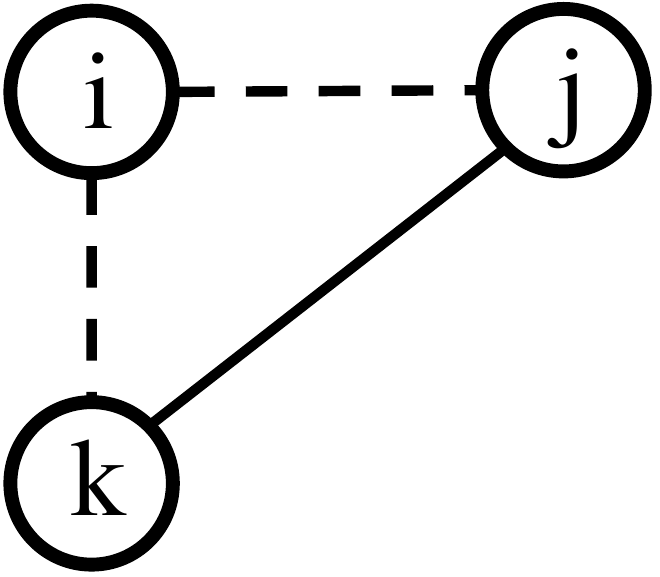}}
		\caption*{Case (2a)} 	\label{2a}
	\end{minipage}
	\hfill 
	\begin{minipage}[c]{0.7\linewidth}
		
		$L =  w_{ij}^+ + w_{jk}^+ + w_{ik}^- = (1-\beta) + \mu$.
		
		\begin{align*}	
		c_{ij}^{} &+ c_{jk}^{} + c_{ik}^{} \\
		&= \beta(1-x_{ij}) + (1-\beta)x_{jk}^{} + \mu(1-x_{ik})\\
		&\geq \beta(1-\delta) + \mu (1-2\delta).
		\end{align*}
		
	\end{minipage}
	
\end{figure}

\begin{figure}[h!]
	\begin{minipage}[c]{0.2\linewidth}
		\centering
		{\includegraphics[width=\linewidth]{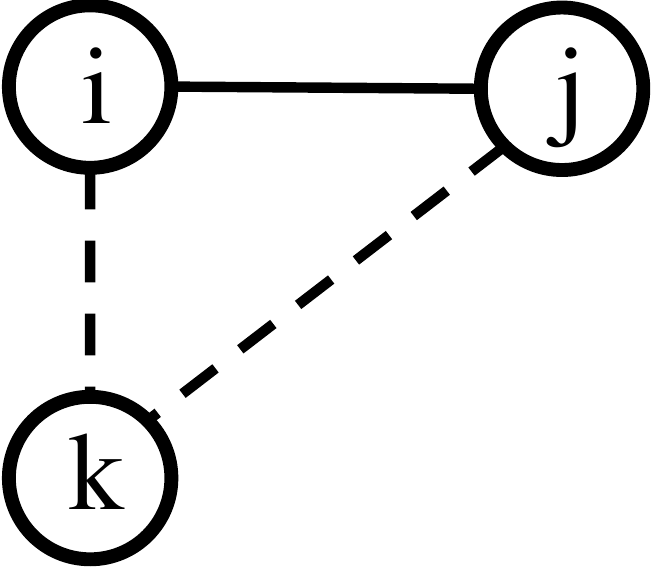}}
		\caption*{Case (2b)}	\label{2b}
	\end{minipage}
	\hfill 
	\begin{minipage}[c]{0.7\linewidth}
		
		$L =  w_{ij}^+ + w_{jk}^+ + w_{ik}^- = (1-\beta) + \mu$.\newline
		
		Observe that this case is symmetric to Case (2a), since the pair $(i,j)$ shares a symmetric relationship to pair $(j,k)$ in the bad triangle.

	\end{minipage}
	
\end{figure}

\begin{figure}[h!]
	\begin{minipage}[c]{0.2\linewidth}
		\centering
		{\includegraphics[width=\linewidth]{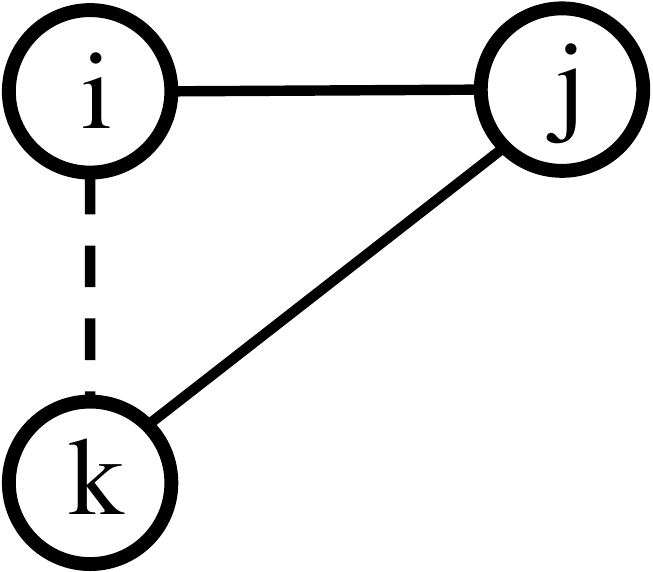}}
		\caption*{Case (2c)} 	\label{2c}
	\end{minipage}
	\hfill 
	\begin{minipage}[c]{0.7\linewidth}
		
		$L =  w_{ij}^+ + w_{jk}^+ + w_{ik}^- = 2(1-\beta) + \mu$.
		
		\begin{align*}	
		c_{ij}^{} + c_{jk}^{} + c_{ik}^{} &= (1-\beta)(x_{ij}^{} + x_{jk}) + \mu(1-x_{ik})\\
		&\geq (1-\beta)x_{ik}^{} + \mu(1-x_{ik})\\
		& \geq (1-\beta)\delta + \mu(1-2\delta).
		\end{align*}
		
	\end{minipage}
	
\end{figure}

\begin{figure}[h!]
	\begin{minipage}[c]{0.2\linewidth}
		\centering
		{\includegraphics[width=\linewidth]{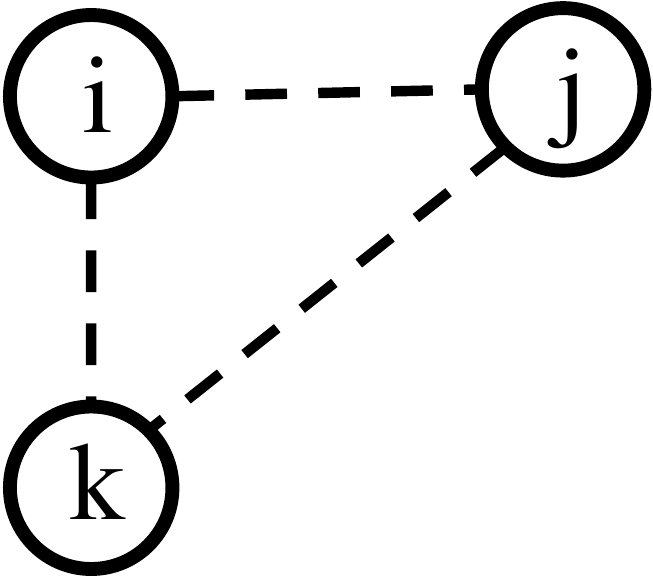}}
		\caption*{Case (2d)} 	\label{2d}
	\end{minipage}
	\hfill 
	\begin{minipage}[c]{0.7\linewidth}
		
		$L =  w_{ij}^+ + w_{jk}^+ + w_{ik}^- = \mu$.
		
		\begin{align*}	
		c_{ij}^{} + c_{jk}^{} + c_{ik}^{} &= \beta(1-x_{ij}^{} + 1-x_{jk}) + \mu(1-x_{ik})\\
		&\geq \beta(2- 2\delta) +\mu(1-2\delta).
		\end{align*}
		
	\end{minipage}
	
\end{figure}

\begin{figure}[t]
	\begin{minipage}[c]{0.2\linewidth}
		\centering
		{\includegraphics[width=\linewidth]{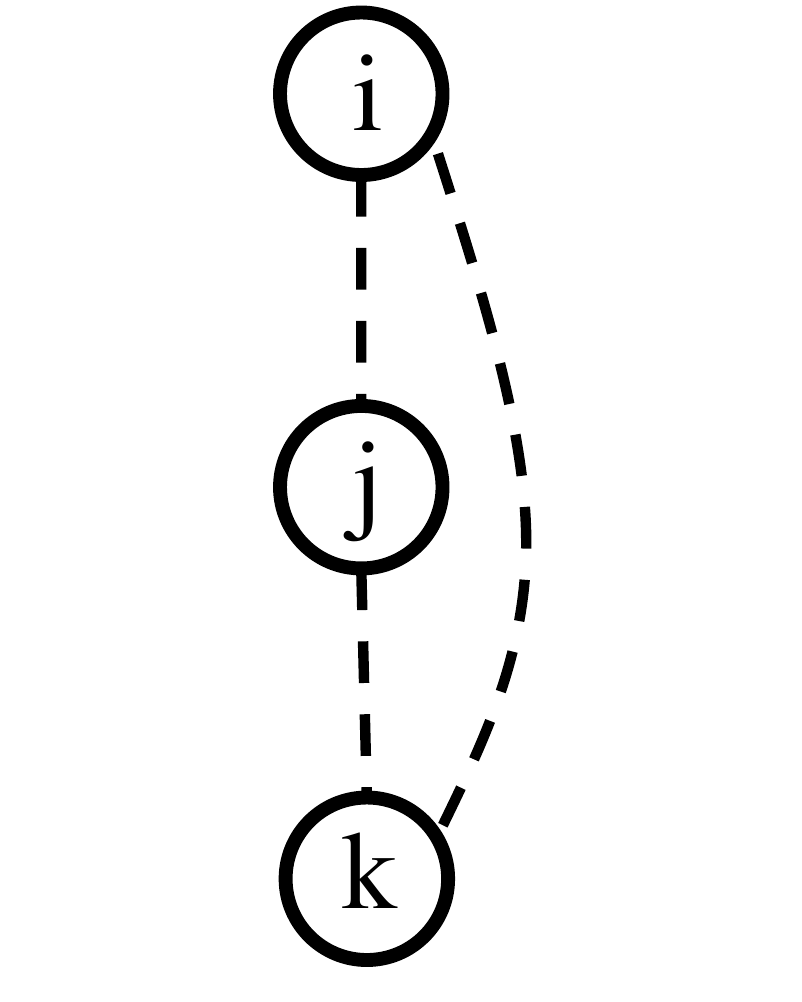}}
		\caption*{Case (3)}	\label{case3}
	\end{minipage}
	\hfill 
	\begin{minipage}[c]{0.7\linewidth}
		
		$L =  w_{ij}^+ + w_{jk}^+ + w_{ik}^- = \mu$.
		
		\begin{align*}	
		c_{ij}^{} + c_{jk}^{} + c_{ik}^{} &= \mu(1-x_{ij}^{} + 1-x_{jk}^{} + 1-x_{ik}) \\
		&\geq \mu(3- 4\delta).
		\end{align*}
		
	\end{minipage}
	
\end{figure}

\end{document}